\documentclass[3p,11pt,letterpaper]{elsarticle}

\usepackage{algorithm} 
\usepackage{algpseudocode}
\usepackage{graphicx} 
\usepackage{amsfonts} 
\usepackage{amsmath}
\usepackage{amssymb} 
\usepackage{amsthm}
\usepackage{bm} 
\usepackage{amsbsy}
\usepackage{url}
\usepackage{color} 
\usepackage[T1]{fontenc} 
\usepackage[ansinew]{inputenc}
\usepackage{lmodern}
\usepackage{todonotes}
\usepackage{enumerate}
\usepackage[colorlinks=true, citecolor=blue,filecolor=black,
  linkcolor=red, urlcolor=blue, hypertexnames=false,pdftex]
  {hyperref}
\usepackage{geometry}

\newtheorem{theorem}{Theorem}[section]

\newtheorem{definition}[theorem]{Definition}
\newtheorem{lemma}[theorem]{Lemma}
\newtheorem{remark}[theorem]{Remark}
\newtheorem{example}[theorem]{Example}

\newcommand{\keyname}{\@capital{data-informed}}

\newcommand{\state}{u}
\newcommand{\amodel}{A}
\newcommand{\lmodel}{L}
\newcommand{\nmodel}{f}
\newcommand{\omodel}{C}
\newcommand{\forward}{F}
\newcommand{\forwardr}{\tilde{F}}
\newcommand{\forwarddr}{\hat{F}}
\newcommand{\misfit}{\eta}
\newcommand{\misfitr}{\tilde{\eta}}

\newcommand{\param}{x}
\newcommand{\data}{y_{\rm obs}}
\newcommand{\mout}{y}
\newcommand{\error}{e}

\newcommand{\hessian}{H}
\newcommand{\map}{\param_{\rm MAP}}

\newcommand{\paramspace}{\mathbb{X}}
\newcommand{\statespace}{\mathbb{U}}
\newcommand{\dataspace}{\mathbb{Y}}
\newcommand{\real}{\mathbb{R}}

\newcommand{\cs}{\mathbb{X}_\perp}

\newcommand{\post}{\pi}
\newcommand{\prior}{\pi_0}
\newcommand{\lkd}{\mathcal{L}}
\newcommand{\papost}{\hat{\pi}}
\newcommand{\japost}{\tilde{\pi}}
\newcommand{\lap}{\pi_{\rm L}}

\newcommand{\ray}{\mathcal{R}}

\newcommand{\normal}{\mathcal{N}}

\newcommand{\prmean}{\mu_{\rm pr}}
\newcommand{\prcov}{\Gamma_{\rm pr}}

\newcommand{\pocov}{\Gamma_{\rm pos}}
\newcommand{\obscov}{\Gamma_{\rm obs}}

\newcommand{\prinner}{\langle \cdot , \cdot \rangle_{\prcov}}
\newcommand{\obsinner}{\langle \cdot , \cdot \rangle_{\obscov}}

\newcommand{\sbasis}{V}

\newcommand{\fbasis}{\Theta}
\newcommand{\mask}{P}
\newcommand{\dbasis}{Y}
\newcommand{\plocalf}{\varphi}
\newcommand{\plocale}{\gamma}
\newcommand{\pbasis}{\Phi}
\newcommand{\pglobalf}{\phi}
\newcommand{\pglobale}{\lambda}
\newcommand{\projlis}{\Pi_r}
\newcommand{\projcs}{\Pi_\perp}

\newcommand{\diag}{{\rm diag}}

\newcommand{\expect}{\mathbb{E}}

\newcommand{\meter}{{\rm m}}
\newcommand{\unittime}{{\rm day}}
\DeclareMathOperator{\range}{range}

\DeclareMathOperator*{\argmax}{arg \, max}

\DeclareMathOperator{\erf}{erf}
\DeclareMathOperator{\rank}{rank}
\DeclareMathOperator{\vspan}{span}
\DeclareMathOperator{\mydim}{dim}

\newcommand{\klpod}{\texttt{Prior-KL-POD}} 
\newcommand{\prdi}{\texttt{Prior-Joint}}
\newcommand{\lapdi}{\texttt{Laplace-Joint}}
\newcommand{\podi}{\texttt{Posterior-Joint}}

\newcommand{\updated}[1]{#1}

\graphicspath{{figures/}}

\journal{Journal of Computational Physics}

\begin{document}


\begin{frontmatter}

  \title{Scalable posterior approximations for large-scale Bayesian
    inverse problems via likelihood-informed parameter and state
    reduction}
    
  \author[mit]{Tiangang Cui}\ead{tcui@mit.edu}

  \author[mit]{Youssef Marzouk\corref{cor1}}\ead{ymarz@mit.edu}
   
  \author[mit]{Karen Willcox}\ead{kwillcox@mit.edu}

  \address[mit]{Massachusetts Institute of Technology, Cambridge, MA
    02139, USA}

  \cortext[cor1]{Corresponding author}

  \begin{abstract}
    Two major bottlenecks to the solution of large-scale Bayesian
    inverse problems are the scaling of posterior sampling algorithms
    to high-dimensional parameter spaces and the computational cost of
    forward model evaluations. Yet incomplete or noisy data, the state
    variation and parameter dependence of the forward model, and
    correlations in the prior collectively provide useful structure
    that can be exploited for dimension reduction in this
    setting---both in the parameter space of the inverse problem
    \textit{and} in the state space of the forward model.
    To this end, we show how to jointly construct low-dimensional
    subspaces of the parameter space and the state space in order to
    accelerate the Bayesian solution of the inverse problem.
    \updated{As a byproduct of state dimension reduction, we also
      show how to identify low-dimensional subspaces of the data in
      problems with high-dimensional observations.}
    These subspaces enable approximation of the posterior as a product
    of two factors: (i) a projection of the posterior onto a
    low-dimensional parameter subspace, wherein the original
    likelihood is replaced by an approximation involving a reduced
    model; and (ii) the marginal prior distribution on the
    high-dimensional complement of the parameter subspace.
    We present and compare several strategies for constructing these
    subspaces using only a limited number of forward and adjoint model
    simulations. \updated{The resulting posterior approximations can
      rapidly be characterized using standard sampling techniques,
      e.g., Markov chain Monte Carlo.}
    Two numerical examples demonstrate the accuracy and efficiency of
    our approach: inversion of an integral equation in atmospheric
    remote sensing, where the data dimension is very high; and the
    inference of a heterogeneous transmissivity field in a groundwater
    system, which involves a partial differential equation forward
    model with high dimensional state and parameters.

  \end{abstract}

  
\end{frontmatter}


\section{Introduction}
\label{sec:intro}

Inverse problems convert indirect observations into useful
characterizations of the parameters of a physical system. These
parameters are related to the observations by a forward model, which
often is expressed as a system of ordinary or partial differential
equations (PDEs) or as an integral equation.
Observations are inevitably corrupted by noise, and the unknown model
parameters may be high-dimensional or infinite-dimensional in
principle. Solution of the inverse problem is thus classically
ill-posed: many feasible realizations of the parameters may be
consistent with the data, and small perturbations in the data may lead
to large perturbations in unregularized parameter estimates.
Rather than seeking regularized point estimates, the Bayesian approach
\cite{IP:Tarantola_2004, IP:KaiSo_2005, IP:Stuart_2010} casts the
inverse solution as the \textit{posterior probability distribution} of
the model parameters conditioned on data, and introduces
regularization in the form of prior information. It thus provides a
means of combining prior knowledge, the data and forward model, and a
stochastic description of measurement and/or model errors; the result
is a principled quantification of uncertainty in parameters and in
parameter-dependent predictions. Characterizing the posterior,
however, is in general a computationally challenging task. The
workhorses of Bayesian computation in this context are Markov chain
Monte Carlo (MCMC) methods
\cite{MCMC:GRP_1996,MCMC:Liu_2001,MCMC:BGJM_2011}, originating with
the Metropolis-Hastings algorithm
\cite{MCMC:Metropolis_etal_1953,MCMC:Hastings_1970}.

A central challenge in the application of MCMC methods to inverse
problems is poor scaling of computational effort with
\textit{parameter dimension} and with the \textit{size of the forward
  model}. High-dimensional parameters frequently represent the
discretization of a spatial field (e.g., the permeability field of a
porous medium) that is the target of inference. Yet the efficiency of
many standard MCMC methods degrades with parameter dimension
\cite{MCMC:RGG_1997, MCMC:RoRo_1998, MCMC:RoRo_2001, MCMC:MPS_2012,
  MCMC:PST_2012}; longer mixing times for MCMC chains then demand more
posterior evaluations to estimate posterior expectations with any
given accuracy. Similarly, many forward models of interest have
high-dimensional states---resulting, for instance, from finite-element
discretizations of PDEs, where many degrees of freedom are needed to
resolve the relevant physics accurately. The {computational expense}
of each forward model evaluation scales {at least} linearly
with the dimension of model state (e.g., when the solution of a linear
system is required). Another important but often-neglected
computational expense in MCMC methods is the proposal process: the
cost of generating random variables and calculating the candidate step
often scales at least linearly with the parameter dimension.

To overcome these twin challenges---parameter dimension and forward
model cost---this paper proposes a likelihood-informed approach for
identifying and exploiting low-dimensional structure in both the
parameter space and the model state space of inverse problems.
Our approach integrates and extends two lines of research: the
likelihood-informed \textit{parameter} dimension reduction of
\cite{DimRedu:Spantini_etal_2015, DimRedu:Cui_etal_2014} and the
data-driven \textit{model} reduction of \cite{ROM:CMW_2014}.
By simultaneously considering the limited accuracy or influence of the
observations, the smoothing properties of the forward model, and the
covariance structure of the prior, the former identifies a
low-dimensional {\it likelihood-informed parameter subspace}
(LIPS)\footnote{\cite{DimRedu:Cui_etal_2014} and \cite{MCMC:CLM_2014}
  refer to this reduced parameter subspace as the
  ``likelihood-informed subspace'' or LIS. Since the present work
  introduces low-dimensional subspaces of both the parameter space and
  state space, we replace `LIS' with the more specific `LIPS' in order
  to avoid confusion.}
where the influence of the likelihood on the posterior dominates that
of the prior.
Given this subspace, one can approximate the full
posterior\footnote{The term ``full posterior'' refers to the posterior
  distribution induced by the full forward model defined on the
  original parameter space.} as the product of a low-dimensional
posterior on the LIPS and the marginalization of the prior onto the
complement of the LIPS.
The latter term can be characterized analytically or with perfectly
independent samples.
Evaluation of the posterior density restricted to the LIPS is still
computationally expensive, however, as it involves the full forward
model. Our second step accelerates these evaluations by projecting
\cite{ROM:NAP_1981, ROM:Sirovich_1987, ROM:HLB_1996, ROM:PaRo_2007}
the forward model---with input parameters restricted to the
LIPS---onto a low-dimensional state subspace. This subspace is called
the {\it likelihood-informed state subspace} (LISS), as it captures
variations in the model state associated with the LIPS-projected
posterior. This model reduction approach extends the data-driven model
reduction ideas of \cite{ROM:CMW_2014} by not only exploiting
posterior concentration, but also avoiding consideration of input
parameter directions that ultimately will not be data-informed.
Finally, we combine these approximations together: the reduced-order
model resulting from projection onto the LISS is substituted into the
product-form approximation of the posterior described above. The
resulting {\it jointly-approximated posterior} is inexpensive to
evaluate and to sample, with a computational cost that is independent
of the dimension of the full model state or the
parameters. \updated{Indeed, this cost scales only with the dimensions
  of the LIPS and the LISS, which are in a sense the \textit{intrinsic
    dimensions} of the problem.}

\updated{As a byproduct of state reduction, we describe new approaches
  for efficiently handling and reducing \textit{high-dimensional data}
  sets in inverse problems; these approaches are potentially useful in
  ``big data'' settings.  While the jointly-approximated posterior
  achieves excellent accuracy in our numerical examples, we also
  discuss how to use it as a proposal distribution in importance
  sampling or delayed-acceptance MCMC \cite{MCMC:ChriFox_2005,
    MCMC:Cui_2010} for the purpose of ``exact'' sampling---i.e., the
  computation of expectations with respect to the full posterior---if
  desired. To ensure convergence in this setting, we introduce a
  special treatment of the tails of the jointly-approximated
  posterior.}

Previous work has also investigated the idea of combining parameter
reduction with model reduction or other forms of surrogate modeling in
order to approximate posterior distributions.
One early effort is \cite{SuMo:MarNa_2009}, which constructs a reduced
parameter basis using the truncated Karhunen-L\`{o}eve (KL) expansion
\cite{DimRedu:Karhunen_1947,DimRedu:Loeve_1978} of the prior
covariance, and then uses generalized polynomial chaos expansions
\cite{SuMo:GhaSpa_1991,SuMo:XiuKar_2002,SuMo:Xiu_2010} to build a
surrogate of the full model.
The same KL-based parameter reduction technique has also been combined
with projection-based model reduction to accelerate posterior
evaluations; examples include \cite{ROM:LSK_2013} and
\cite{ROM:CMW_2014}. Similarly, \cite{SuMo:MaZa_2009} combines a
process convolution model \cite{Sto:Higdon_2002} of the parameters
with a sparse grid approximation of the forward model.
A different approach in \cite{ROM:LWG_2010} simultaneously identifies
reduced subspaces for both the parameters and the state, by solving a
sequence of model-constrained optimization problems penalized by the
prior.
In general, all these earlier approaches seek a truncation of the
parameter dimension and then accelerate forward model evaluations over
the reduced parameter subspace using surrogates.
The smoothness of the prior plays a crucial role in these approaches,
either for avoiding large KL truncation errors or for promoting
convergence of the model-constrained optimization.
In practice, this requirement can impose significant restrictions on
the choice of priors.
Our approach is fundamentally different in several respects.
First, our posterior approximation is based on capturing the
\textit{change} from the prior to the posterior within the LIPS,
rather than directly truncating the parameter dimension of the
problem.
Second, model reduction using the LISS only captures the state
variations that are relevant to the change from prior to posterior;
this ``localization'' strategy is key to the successful construction
of reduced-order models for high-dimensional parameterized systems.
Furthermore, our jointly-approximated posterior is not singular with
respect to the full posterior, and can thus be used to drive exact
sampling schemes (e.g., importance sampling as mentioned above).

The rest of this paper is organized as follows.
In Section \ref{sec:background}, we review the Bayesian formulation of
inverse problems.
In Section \ref{sec:methodology}, we introduce the concept of joint
posterior approximation using reduced parameter and state subspaces,
then detail various strategies and practical algorithms for
constructing this approximation and for exploring the full posterior.
In Section \ref{sec:gomos}, we demonstrate various aspects of our
proposed approach using an atmospheric remote sensing problem,
comparing different strategies for subspace identification and for the
reduction of high-dimensional data.
In Section \ref{sec:elliptic}, we apply our joint posterior
approximation approach to the inference of the transmissivity field of
a groundwater aquifer.
Section \ref{sec:conclusions} offers concluding remarks.
  



\section{Bayesian formulation for inverse problems}
\label{sec:background}

We begin by constructing the forward model. Consider a numerical
discretization of the system of interest, described by a nonlinear
equation
\begin{equation}
  \amodel(\state, \param) = 0,
  \label{eq:amodel}
\end{equation}
where $\state \in \statespace \subseteq \real^{m}$ and
$\param \in \paramspace \subseteq \real^{n}$ are the $m$-dimensional
state vector and the $n$-dimensional parameter vector, respectively.
The goal of an inverse problem is to infer the unobservable parameters
$\param$ from noisy partial observations of the states $\state$, given by
\begin{equation}
  \data = \omodel(\state, \param) + \error \, .
  \label{eq:omodel}
\end{equation}
Here $\omodel$ is a discretized observation operator mapping from the
states and parameters to the observables, and $\error$ is a random
variable representing noise and/or model error, which appear additively.
The system model $\amodel(\state, \param) = 0$ and observation model
$\omodel(\state, \param)$ together define a forward model
$\mout = \forward(\param)$ that maps the unknown parameter to the
observable model outputs.
We note that although the forward model defined by \eqref{eq:amodel}
and \eqref{eq:omodel} is induced by a stationary problem, the
methodology presented in this paper is also applicable to
time-dependent systems.%

To formulate the inverse problem in a Bayesian setting, we model the
parameter $\param$ as a random variable, endow it with a prior
distribution, and then characterize its posterior distribution
given a realization of the data,
$\data \in \dataspace \subseteq \real^{d}$:
\begin{equation}
  \label{eq:post}
  \post(\param \vert \data) \propto \lkd( \data \vert \param )  \prior(\param).
\end{equation}
Here, we assume that all distributions have densities with respect to Lebesgue measure.
The posterior density above is the product of two terms: the prior
density $\prior(\param)$, which models knowledge of the parameters
before the data are observed, and the likelihood function
$ \lkd( \data \vert \param )$, which describes the probability
distribution of $\data$ for a given $\param$.

We develop our formulation in the setting of a multivariate Gaussian
prior $\normal(\prmean, \prcov)$, where the covariance matrix $\prcov$
might also be specified by its inverse $\prcov^{-1}$, commonly
referred to as the precision matrix. The additive observational noise
is taken to be a zero mean Gaussian distribution, i.e.,
$\error \sim \normal(0, \obscov)$.
Given the weighted inner product
$\langle \mout_1 , \mout_2 \rangle_{\obscov} = \langle \mout_1 ,
\obscov^{-1} \mout_2 \rangle$
and the induced norm
$ \| \mout \|_{\obscov} = \sqrt{\langle \mout , \mout
  \rangle_{\obscov}} $, we can define a data-misfit function
\begin{equation}
  \misfit(\param) = \frac{1}{2}  \left\|  \forward(\param) - \data  \right\|^2_{\obscov}.
  \label{eq:misfit}
\end{equation}
The likelihood function is thus proportional to $\exp \left( -\misfit(\param) \right )$.
The Gaussian settings used here can be generalized to non-Gaussian
priors, e.g., log-normal distributions, with an appropriate
transformation or change of variables to a Gaussian. Additive but
non-Gaussian noise can be handled similarly. These transformations may
introduce additional nonlinearity in the forward model.

Note that the unknown parameters and the model states are in principle
functions of space and/or time, and that the finite-dimensional
representations above are the result of numerical discretization.
If one considers progressively refining the parameter discretization,
however, the posterior distribution does not have a density with
respect to Lebesgue measure at the infinite-dimensional limit.
However, for inverse problems with properly chosen Gaussian
priors---e.g., a covariance operator that is self-adjoint, positive
definite, and trace-class---and a forward model satisfying certain
regularity conditions---e.g., appropriately bounded
\cite{IP:Stuart_2010} and locally Lipschitz---the posterior has a
density with respect to the prior and yields a full measure at the
infinite-dimensional limit.
In this case, Bayes' rule in \eqref{eq:post} is expressed as the
Radon-Nikodym derivative of the posterior with respect to the
prior. We refer the readers to \cite{IP:Stuart_2010} and references
therein for more details.
Since we aim to approximate the posterior distribution defined by
given discretizations of the parameters, a finite-dimensional
forward model, and the associated prior, we adopt the
finite-dimensional representation of the posterior as our starting
point in this paper.
This finite-dimensional posterior can be derived either from a
consistent discretization of an infinite-dimensional inverse problem
or from some other existing numerical models that are not necessarily
well-defined in the infinite-dimensional limit.



\section{Posterior approximation via dimension reduction}
\label{sec:methodology}

In this section, our first objective is to reduce the algorithmic
complexity of posterior sampling by identifying a likelihood-informed
parameter subspace (LIPS) that captures parameter directions where the
change from prior to posterior is most significant. 
We will then decompose the posterior into the product of (i) a
low-dimensional distribution, defined on the LIPS, that is analytically
intractable and therefore must be sampled; and (ii) a higher-dimensional
but analytically tractable marginal of the prior distribution on the
complement of the LIPS.
To accelerate the forward model evaluations required when sampling the
first term of this product decomposition, we will identify a
low-dimensional subspace of the forward model state---the
likelihood-informed state subspace (LISS)---and construct a reduced
version of the forward model accordingly. Introducing this reduced
model into the product decomposition yields the jointly-approximated
posterior distribution described in the introduction. 
At the end of this section, we will describe and compare various
sampling strategies for constructing the LISS and LIPS, and hence the
jointly-approximated posterior. We will also discuss methods for
exploring the full posterior distribution, given the preceding
approximations.


\subsection{Data-informed parameter reduction}
\label{sec:lips}

Inverse problems very often involve some combination of a smoothing
forward model, limited or noisy observations, and correlations in the
prior. When any of these factors is present, the data will not equally
inform all directions in the parameter space. We may be able to
explicitly project the argument of the likelihood function onto a
lower-dimensional subspace of the parameter space without losing much
information.
Our objective here is to find an $r$-dimensional LIPS, denoted by
$\paramspace_r$, to capture the parameter directions where the
likelihood is ``most informative'' relative to the prior. This notion will
be defined more precisely below.

\subsubsection{Parameter-reduced posterior}

Consider a rank-$r$ projector $\projlis$ whose range is the LIPS, i.e.,
$\paramspace_r = \range(\projlis)$. We approximate the
posterior density \eqref{eq:post} as:
\begin{equation}
  \label{eq:approx_post_param_0}
  \papost(\param \vert \data) 
\propto  \lkd\left(\data  \vert  \projlis \param\right) \prior(\param), 
\end{equation}
where $ \lkd\left(\data \vert \projlis \param\right)$ is an
approximation to the original likelihood function
$ \lkd \left ( \data \vert \param \right )$.
We require the projector $\projlis$ to be orthogonal with respect to
the inner product induced by the prior covariance
$\langle \param_1 , \param_2 \rangle_{\prcov} = \langle \param_1 ,
\prcov^{-1} \param_2 \rangle$.
This requirement is essential to constructing a tractable product-form
approximation of the posterior.

\begin{definition}[Parameter space projectors]
  \label{def:proj}
  Suppose the subspace $\paramspace_r$ is spanned by a basis
  $\pbasis_r \in \real^{n \times r}$ that is orthogonal with respect
  to the inner product $\prinner$, i.e.,
  $\langle \pbasis_r , \pbasis_r \rangle_{\prcov} = I_r$, where $I_r$
  is the $r$-by-$r$ identity matrix.
  Define the matrix $\Xi_r \in \real^{n \times r}$ such that
  $\Xi_r^\top \pbasis_r^{} = I_r^{}$. Then the projectors
  \[
  \projlis^{} = \pbasis_r^{} \Xi_r^\top \quad {\rm and} \quad \projcs
  = I - \projlis,
  \]
  are orthogonal with respect to $\prinner$.
  Moreover, the $(n-r)$-dimensional subspace
  $\paramspace_\perp = \range(\projcs)$ is the orthogonal complement
  of $\paramspace_r$ with respect to $\prinner$.
  We can choose a basis $\pbasis_\perp \in \real^{ n \times (n-r) }$
  such that $[ \pbasis_r , \pbasis_\perp ]$ forms a complete
  orthogonal system in $\real^n$ 
  with respect to $\prinner$, and thus the projector
  $\projcs$ can be written as
  $\projcs^{} = \pbasis_\perp^{} \Xi_\perp^\top$, where
  $ \Xi_\perp \in \real^{ n \times (n-r) }$ is the matrix such that
  $\Xi_\perp^\top \pbasis_\perp^{} = I_\perp^{}$.
\end{definition}

Using the projectors defined above, the parameter $\param$ can be
decomposed as $ \param = \projlis \param + \projcs \param $,
where each projection can be represented as the linear combination of
the corresponding basis vectors.
Consider parameters $\param_r$ and $\param_\perp$ that are the weights
associated with the bases $\pbasis_r$ and $\pbasis_\perp$,
respectively.
Then we can define the following pair of linear transformations
between $\param$ and $(\param_r, \param_\perp)$:
\begin{equation}
  \param = \left[\pbasis_r \; \pbasis_\perp\right]
  \left[\begin{array}{l} \param_r \\ \param_\perp \end{array}\right]
  \quad {\rm and} \quad
  \left[\begin{array}{l} \param_r \\ \param_\perp \end{array}\right] = \left[\Xi_r \; \Xi_\perp \right]^\top \param.
  \label{eq:decom}
\end{equation}

\begin{lemma}
  \label{lemma:prior}
  Given the decomposition
  $\param = \pbasis_r \param_r + \pbasis_\perp \param_\perp$ defined in
  \eqref{eq:decom}, the prior distribution can be decomposed into the
  product form
  $\prior(\param) \propto \prior(\param_r) \prior(\param_\perp)$,
  where $\prior(\param_r) = \normal(\Xi_r^\top \prmean, I_r^{})$ and
  $\prior(\param_\perp) = \normal(\Xi_\perp^\top \prmean,
  I_\perp^{})$.
\end{lemma}

Applying the linear transformation \eqref{eq:decom} and Lemma
\ref{lemma:prior}, the \textit{parameter-approximated} posterior
\eqref{eq:approx_post_param_0} can be written as
\begin{equation}
  \papost(\param \vert \data) \propto \papost(\param_r, \param_\perp \vert \data) = \post(\param_r \vert \data) \prior(\param_\perp),
  \label{eq:approx_post_param}
\end{equation} 
which is the product of the {\it parameter-reduced} posterior
\begin{equation}
  \post(\param_r \vert \data) \propto \lkd\left(\data \vert \pbasis_r \param_r\right) \prior(\param_r),
  \label{eq:redu_post_param}
\end{equation}
and the {\it complement prior} $\prior(\param_\perp)$.
In this product-form approximation, the prior-to-posterior update is
captured entirely by the parameter-reduced posterior.
Approximations of this form naturally yield a scalable posterior exploration scheme: the high-dimensional complement
prior $\prior(\param_\perp)$ is analytically tractable, and the
remaining challenge is to explore the analytically intractable but
low-dimensional parameter-reduced posterior
\eqref{eq:redu_post_param}.
Of course, this construction rests on identifying the LIPS basis
$\pbasis_r$. In the rest of this subsection, we will discuss several
ways to construct the LIPS by balancing the influence of the prior and
the likelihood.

\begin{remark}
  It is usually not feasible to compute and store the high-dimensional
  basis $\pbasis_\perp \in \real^{n \times (n-r)}$.
  In fact, the construction of the parameter-approximated posterior
  \eqref{eq:approx_post_param} only requires the low-dimensional LIPS
  basis $\pbasis_r$ in order to construct the parameter-reduced posterior.
  Operations involving the complement subspace $\cs$ are performed
  using the projector $\projcs = I - \projlis$ rather than the
  high-dimensional basis $\pbasis_\perp$.
\end{remark}

\subsubsection{Likelihood-informed parameter subspace}

In previous work \cite{DimRedu:Spantini_etal_2015,IP:Flath_etal_2011,
  MCMC:MWBG_2012, IP:BGMS_2013, MCMC:Petra_etal_2014}, the Hessian of
the data-misfit function \eqref{eq:misfit} (in particular, the Gauss-Newton
approximation $\hessian(\param)$ of the Hessian) is used to quantify
the local impact of the likelihood, relative to the prior, along a
parameter direction $\plocalf$. This notion of relative impact is
captured via the local Rayleigh ratio
\begin{equation}
  \ray(\plocalf;\param) = \frac{\langle \plocalf, \hessian(\param) \plocalf\rangle}{
    \langle \plocalf, \prcov^{-1} \plocalf \rangle },
  \label{eq:ray_local}
\end{equation}
which is maximized (over successively smaller subspaces
$\vspan^\bot{(\plocalf_j)_{j < i}}$) by generalized eigenvectors of the matrix pencil $ (\hessian(x),
  \prcov^{-1} )$
\begin{equation}
  \hessian(\param) \plocalf_i^{} = \plocale_i^{} \prcov^{-1} \plocalf_i^{}.
  \label{eq:eig_local}
\end{equation}
The largest eigenvalues of \eqref{eq:eig_local} correspond to
parameter directions along which the local curvature of the
log-posterior density is more constrained by the log-likelihood than
by the log-prior. Conversely, eigendirections associated with the
smallest eigenvalues correspond to directions along which the
likelihood is essentially flat, and hence where the posterior is (locally)
determined by the prior.

Given the local Gauss-Newton approximation of the Hessian,
$\hessian(\param)$, we can define a local Gaussian approximation of
the posterior with covariance $\pocov(\param) := (\hessian(\param) +
\prcov^{-1})^{-1}$. This covariance can be
written as a low-rank update of the prior covariance:
\begin{equation}
  \pocov^{}(\param) \approx \pocov^{(l)} =  \prcov^{} - \sum_{i = 1}^{l} \frac{\plocale_i^{}}{\plocale_i^{}+1}\plocalf_i^{} \plocalf_i^\top \, ,
  \label{eq:local_gauss}
\end{equation}
which is in general approximate for $l < n$. Here $\plocalf_i$ and
$\plocale_i$ are eigenvectors and eigenvalues from \eqref{eq:eig_local},
which depend on the parameter $\param$.
In this approximation, the basis $\{\plocalf_1, \ldots, \plocalf_l\}$
characterizes the prior-to-posterior update at a given $\param$.

This low-rank update was first used in \cite{IP:Flath_etal_2011} for
computing and factorizing the posterior covariance in large-scale
linear inverse problems.
Spantini et al.\ \cite{DimRedu:Spantini_etal_2015} proved the \textit{optimality} of
this approximation at any given $l$, for linear inverse problems, in
the sense of minimizing the F\"{o}rstner-Moonen \cite{Lin:Forbou_2003}
distance from the exact posterior covariance matrix over the class of
positive definite matrices that are rank-$l$ negative semidefinite
updates of the prior covariance.
For nonlinear inverse problems, given the posterior mode (i.e., the
  maximum a posteriori (MAP) estimator)
\[
\map = \argmax_{x} \pi(\param \vert \data),
\]
the local Gaussian approximation \eqref{eq:local_gauss} centered at
 $\map$ yields a Laplace approximation of the posterior:\footnote{More
   precisely, this is a Laplace approximation with an additional
   approximation of the covariance as a low-rank update of the prior,
   for $l < \min \left (\rank(\hessian), \rank(\prcov) \right )$.}
\begin{equation}
  \post(\param \vert \data) \approx \lap(\param) = \normal\left(
    \map\, , \, \prcov^{} - \sum_{I = 1}^{l}
    \frac{\plocale_i^{}}{\plocale_i^{}+1}\plocalf_i^{} \plocalf_i^\top
  \right).
  \label{eq:lap}
\end{equation}
For unimodal and nearly Gaussian posteriors, the Laplace approximation
might be used directly as a surrogate for the posterior, as in
\cite{IP:BGMS_2013}; alternatively, it can be used as a fixed
preconditioner for MCMC sampling, as in the stochastic Newton method of
\cite{MCMC:Petra_etal_2014}. We will use the Laplace approximation as
a component of certain subspace construction strategies, described in
the next subsection.

We wish to construct the parameter-approximated
posterior \eqref{eq:approx_post_param} for nonlinear inverse problems,
and thus must extend beyond Gaussian approximations.
Since the Hessian varies over the parameter space for nonlinear
forward models, the likelihood-informed directions also vary with
$\param$ and are embedded in some nonlinear manifold.
To build a global linear subspace---the LIPS---that captures the
majority of this nonlinear manifold,
Cui et al. \cite{DimRedu:Cui_etal_2014} extends the pointwise criterion
\eqref{eq:ray_local} into the expected value of the Rayleigh quotient
over the posterior
\begin{equation}
  \expect_{\pi(\param \vert \data)}\left[\ray(\pglobalf; \param)\right] =  \frac{\langle \pglobalf, S_{\rm post} \, \pglobalf\rangle}{ \langle \pglobalf, \prcov^{-1} \pglobalf \rangle },
  \label{eq:ray_global}
\end{equation}
where $S_{\rm post}$ is the posterior expectation of the Gauss-Newton approximation of the
Hessian (GNH):
\begin{equation}
  S_{\rm post} = \int_{\paramspace}  \hessian(\param) \, \pi(d\param \vert \data).
  \label{eq:expect_hessian}
\end{equation}
This way, the LIPS can be obtained through the eigendecomposition of
the matrix pencil $\left(S_{\rm post} , \prcov^{-1}\right)$,
\begin{equation}
  S_{\rm post} \, \pglobalf_i^{} = \pglobale_i^{} \prcov^{-1} \pglobalf_i^{}.
  \label{eq:eig_global}
\end{equation}
The eigenvectors $\{\pglobalf_1, \ldots, \pglobalf_r\}$ correspond
to the $r$ leading eigenvalues of \eqref{eq:eig_global}, such that
$\pglobale_1 \geq \pglobale_2 \geq \ldots \geq \pglobale_r \geq
\tau_{\rm g} > 0$, span the LIPS.
Here the truncation threshold $\tau_{\rm g}$ is usually set to a value
less than one, e.g., $\tau_{\rm g} = 10^{-1}$, thus only removing parameter directions
where the impact of the likelihood is significantly smaller than that of the prior.

The evaluation of the expected GNH \eqref{eq:expect_hessian} should be
carried out adaptively during posterior exploration.
In particular, we consider approximating $S_{\rm post} $ using
Monte Carlo integration,
\[
\widehat{ S }_{\rm post} = \frac{1}{N} \sum_{k = 1}^{N}
H(\param^{(k)}),
\]
where $\param^{(k)} \sim \pi(\param \vert \data)$, $k = 1 \ldots N$,
are posterior samples adaptively selected during posterior
exploration. Section \ref{sec:algo} will describe how we fit this task
into an adaptive sampling framework where posterior exploration and
posterior approximation are carried out simultaneously.

\begin{lemma}
  \label{lemma:LIPS}
  The eigenvectors $\{\pglobalf_1, \ldots, \pglobalf_r\}$ are linearly
  independent and form a LIPS basis
  $\pbasis_r = [ \pglobalf_1, \ldots, \pglobalf_r ]$ that is
  orthogonal with respect to the inner product $\prinner$.
\end{lemma}
\begin{proof}
  The result directly follows from the fact that the estimated
  expected GNH $\widehat{ S }_{\rm post} $ is symmetric positive
  semidefinite and the prior covariance $\prcov$ is symmetric positive
  definite. See Theorem 15.3.3 of \cite{Lin:Parlett_1980} for details.
\end{proof}

As a consequence of Lemma \ref{lemma:LIPS}, we can construct the
matrix
$\Xi_r^{} = \pbasis_r^{} \, (\pbasis_r^\top \, \pbasis_r^{})^{-1} $,
such that $\Xi_r^\top \pbasis_r^{} = I_r$, and the projector
$\projlis^{} = \pbasis_r^{} \Xi_r^\top$ as in Definition
\ref{def:proj}.
This way, the parameter-approximated posterior
\eqref{eq:approx_post_param} and parameter-reduced posterior
\eqref{eq:redu_post_param} can be defined.

\subsubsection{Parameter subspace identification: choice of reference distribution}
\label{sec:paramref}

The LIPS basis discussed above results from balancing the
influence of the prior and the likelihood over the support of the
posterior.
It is also possible to quantify this relative influence using
expectations of the local Rayleigh quotient over \textit{other} reference
distributions.
Depending on the choice of reference distribution, the local
likelihood-informed directions---summarized by the local Rayleigh
quotient---are weighed differently in the resulting LIPS.
Furthermore, the involvement of the observed data set in the reference
distribution affects how computational resources might be allocated to
LIPS construction.
If the observed data are not involved in the reference distribution,
the evaluation of the expectation can be performed once for a
particular combination of forward model and prior, and reused for
different data sets.\footnote{The integrand $\hessian(\param)$ does not
  depend on the data, and therefore this expectation does not depend
  on the data.}
This way, LIPS construction is decoupled from any particular data set,
and we consider it to be an {\it offline} procedure.
In contrast, if the reference distribution involves the
observation---e.g., if it is the posterior---the resulting expectation
evaluation is an {\it online} procedure, as the LIPS basis must be
recomputed for each new data set.

One obvious candidate for the reference distribution is the prior,
which leads to a LIPS basis that is constructed from the
eigendecomposition of the matrix pencil
$\left( S_{\rm pr} , \prcov^{-1} \right)$, where
\[
S_{\rm pr} = \int_\paramspace \hessian(\param) \, \prior(d\param),
\]
is the expected GNH over the prior.
As discussed above, using the prior as a reference distribution leads
to an offline procedure for parameter dimension reduction.

Alternatively, we can set the reference distribution to be the Laplace
approximation \eqref{eq:lap}, which is an inexpensive and
easy-to-sample surrogate for the posterior.
This way, the LIPS can be obtained from the eigendecomposition of the
pencil $\left( S_{\rm L} , \prcov^{-1} \right)$, where
\[
S_{\rm L} = \int_\paramspace \hessian(\param) \, \lap(d\param).
\]
Although samples can be directly drawn from the Laplace approximation,
this choice of reference constitutes an online approach, since the
Laplace approximation is data-dependent.

All of the LIPS bases discussed above are orthogonal with respect to
the inner product $\prinner$, since they result from generalized
eigenproblems involving $\prcov^{-1}$.  We note that other choices of
reduced parameter basis can also lead to the product-form approximated
posterior \eqref{eq:approx_post_param}, provided that the basis is
orthogonal with respect to $\prinner$.
For instance, the basis defined by the truncated Karhunen-Lo\`{e}ve
expansion of the prior covariance \cite{SuMo:MarNa_2009} also
satisfies this property.
\newcommand{\spost}{S_{\rm post} = \int_\paramspace \hessian(\param)
  \pi(d\param \vert \data)}
\newcommand{\sprior}{S_{\rm pr}  = \int_\paramspace
  \hessian(\param) \prior(d\param) \; \; \,}
\newcommand{\slap}{S_{\rm L}  = \int_\paramspace
  \hessian(\param) \pi_{\rm L}(d\param) \; \,}
\begin{table}[h!]
  \caption{Summary of parameter dimension reduction methods discussed in
    Section~\ref{sec:paramref}: the associated eigendecomposition,
    whether the process is offline/online, and whether adaptive
    sampling is required.}
  \label{table:param_dim_redu}
  \begin{center}
    \begin{tabular}{l|c|c|l|l}
      \hline 
      Method & \multicolumn{2}{l|}{Eigendecomposition} & Online/offline & Adaptive sampling?  \\
      \hline 
      Posterior-LIPS & $\left( S_{\rm post}, \prcov^{-1} \right)$ & $\spost$ & online & yes\\
      Laplace-LIPS & $\left( S_{\rm L}, \prcov^{-1} \right)$ & $\slap$ & 
                                                                                  online & no\\
      Prior-LIPS & $\left( S_{\rm pr}, \prcov^{-1} \right)$ & $\sprior$ & offline & no\\
      Prior-KL   & 
  $\prcov$  & -- & offline & no\\
      \hline
    \end{tabular}
  \end{center}
\end{table}

\begin{figure}[h!]
  \centerline{\includegraphics[width=\textwidth]{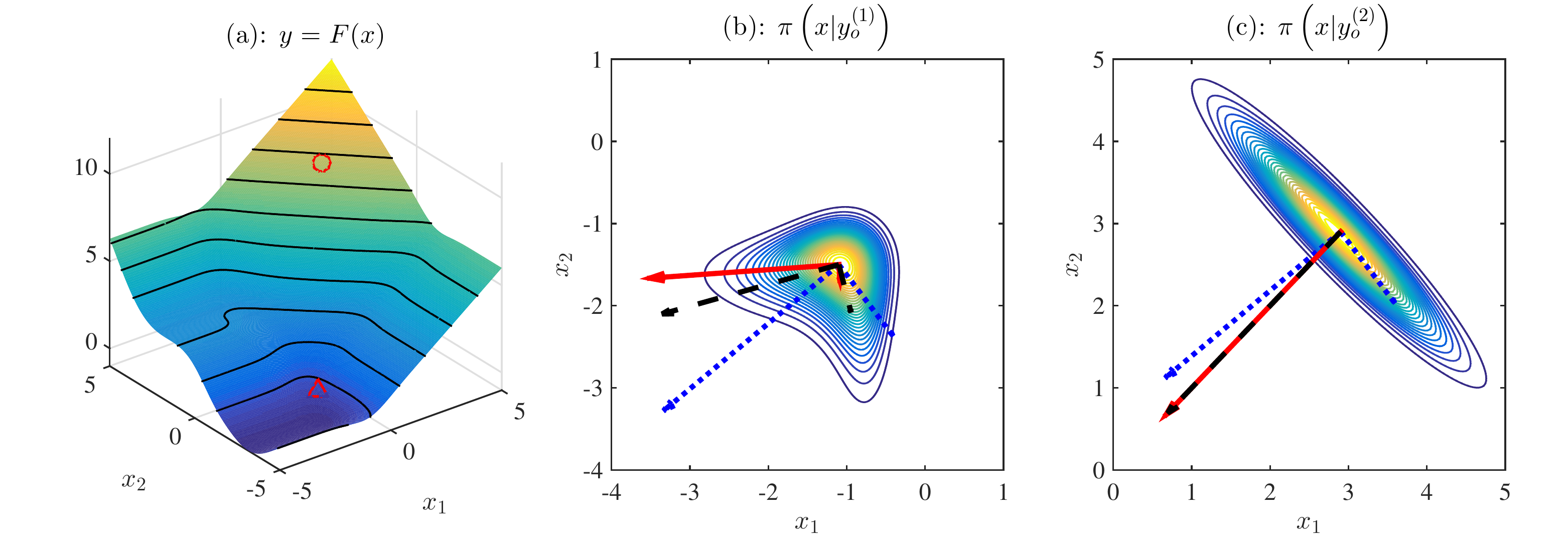}}
  \caption{Illustration of Example~\ref{ex:paramred}. (a) Forward
    model, with two different observations $\data^{(1)}$ and
    $\data^{(2)}$ marked by the triangle and the circle,
    respectively. (b) Contours of the posterior density
    $\pi(\param \vert \data^{(1)} )$. (c) Contours of the posterior
    density $\pi(\param \vert \data^{(2)} )$. On (b) and (c), the
    parameter bases computed using posterior-LIPS, Laplace-LIPS, and
    prior-LIPS are depicted by the solid, dashed, and dotted lines,
    respectively. }
  \label{fig:mot_dips}
\end{figure}

All of the parameter subspace identification techniques presented here can be
interpreted as the result of an eigendecomposition.
For each technique, Table \ref{table:param_dim_redu} summarizes the
specific form of eigendecomposition, whether reduction is
offline/online, and whether adaptive sampling is required.
We compute the expected Hessians required for both Prior-LIPS and
Laplace-LIPS using Monte Carlo integration, where samples can be
directly generated from the prior and the Laplace approximation.
These Hessian evaluations are therefore embarrassingly parallel.
Posterior-LIPS, on the other hand, requires posterior sampling; we
defer a discussion of the associated adaptive framework to
Section~\ref{sec:algo}.
We also note that Prior-KL is less computationally demanding than
the other methods discussed here, as it only involves an
eigendecomposition of the prior covariance.

\begin{example}[Parameter reduction with different reference distributions]
\label{ex:paramred}
To demonstrate the properties of different parameter reduction
approaches, we consider a two-dimensional inverse problem with the
following scalar-valued forward model:
\newcommand{\fbase}[1]{\left(\erf\left( #1+1\right) + 1 \right) \left
    (#1+1 \right)}
\newcommand{\fjoint}[2]{\left( 1-\erf\left( #1 + 1 \right) \right)
  \cos(#2) }
\[
\forward(\param) := \frac12 \left[ \fbase{\param_1} + \fbase{\param_2}
  + \fjoint{\param_1}{\param_2} \right].
\label{eq:mot_lis_f}
\]
%
We define a Gaussian observational error
$\error \sim \normal(0, \sigma^{-2} I_1)$ with $\sigma = 0.25$, and a
prior $\prior(\param) = \normal(0, I_2)$.
Figure \ref{fig:mot_dips}(a) shows the response of the forward model
over a range of input parameter values $(\param_1, \param_2)$.
Different observations will yield posterior distributions with
dramatically different shapes and orientations, due to the
nonlinearity of the forward model.
We consider two different observation values $\data^{(1)}$ and
$\data^{(2)}$, generated by setting the underlying true parameter to
$(-1, -2)$ and $(3, 3)$, respectively.
These observations are illustrated by the triangle and the circle in
Figure \ref{fig:mot_dips}(a).

As shown in Figure \ref{fig:mot_dips}(b), the posterior distribution
$\pi(\param \vert \data^{(1)} )$ has strongly non-Gaussian structure
in both parameter dimensions, as the observation $\data^{(1)}$ falls
in a regime where the forward model exhibits
nonlinear behavior.
In contrast, the observation $\data^{(2)}$ corresponds to a
parameter regime where the forward model is rather linear, and thus
the posterior distribution $\pi(\param \vert \data^{(2)} )$,
shown in Figure \ref{fig:mot_dips}(c), has nearly Gaussian structure.

In this example, the Prior-KL basis simply corresponds to the canonical basis
of the posterior, as the prior is an independent standard Gaussian.
The Prior-LIPS basis is illustrated with dotted lines and remains
unchanged for both data sets, while the Posterior-LIPS basis is
illustrated with solid lines and depends on the data set.
The vectors corresponding to the Prior-LIPS basis and the
Posterior-LIPS basis in the figure are scaled by their associated
generalized eigenvalues.
For the data set $\data^{(2)}$, the likelihood function only informs
\emph{one} dimension in the parameter space, since the forward model is
nearly linear in the support of the posterior.
Posterior-LIPS captures this likelihood-informed subspace accurately,
but Prior-LIPS does not; it suggests that two parameter directions are
of comparable importance.
Laplace-LIPS produces results similar to Posterior-LIPS for
the data set $\data^{(2)}$, but the bases generated by
Laplace-LIPS and Posterior-LIPS are rather different for the data
set $\data^{(1)}$.
This is expected, because the posterior
$\pi(\param \vert \data^{(2)} )$ is nearly Gaussian and well described
by its Laplace approximation, whereas the posterior
$\pi(\param \vert \data^{(1)} )$ is strongly non-Gaussian.
\end{example}

As shown by the above example, all three strategies for computing the
LIPS can reveal the impact of the likelihood on different parameter
directions, relative to the prior.
The Prior-LIPS does not depend on the data set, and hence can be
constructed offline for all possible observations.
Its potential drawback is that the resulting low-dimensional parameter
subspace may contain directions unimportant to the posterior at hand;
alternatively, some important posterior structure may be missed, due
to the averaging of different Hessians over a much less concentrated
parameter measure.
In comparison, the online Posterior-LIPS can accurately capture
likelihood-informed directions for a specific posterior distribution
of interest.
Depending on the shape of the posterior, Laplace-LIPS can be used
either as a replacement for Posterior-LIPS when the posterior is nearly
Gaussian, or as an initial guess for Posterior-LIPS that is adaptively
refined (see Section~\ref{sec:algo}).


\subsection{Likelihood-informed state reduction}
\label{sec:liss}

Although the parameter-approximated posterior
$\papost(\param \vert \data) \propto \post(\param_r \vert \data)
\prior(\param_\perp)$
enables significant reductions in the algorithmic complexity of
posterior exploration, by confining sampling to a lower dimensional
space, a remaining computational bottleneck is the exploration of the
parameter-reduced posterior $\post(\param_r \vert \data)$. The cost of
this exploration is dominated by forward model evaluations, and thus
we turn to model reduction approaches.

\subsubsection{Model reduction}

In the parameter-reduced posterior, by projecting the argument of the
likelihood function onto the subspace spanned by the reduced parameter
basis $\pbasis_r$, the forward model defined by \eqref{eq:amodel} and
\eqref{eq:omodel} can be rewritten as
\begin{equation}
  \amodel(\state, \pbasis_r \param_r) = 0, \quad {\rm and} \quad
  \mout = \omodel(\state, \pbasis_r \param_r) .
  \label{eq:pfmodel}
\end{equation}
To reduce the computational cost, we wish to solve a projection of the
parameter-reduced forward model \eqref{eq:pfmodel} onto a reduced
dimensional state subspace.
Without loss of generality, we consider the system model
$\amodel(\state, \param) = 0$ to consist of a linear operator
$\lmodel$ and a nonlinear function $\nmodel$, which take the form
\begin{equation}
  \label{eq:amodel_LF}
  \lmodel(\param) \state + \nmodel(\param, \state)  = 0.
\end{equation}
Suppose that variation of the $m$-dimensional state lies within an
$s$-dimensional subspace spanned by a basis
$\sbasis_s \in \real^{m \times s}$. Then the state $\state$ can be
approximated by a linear combination of the reduced basis vectors,
i.e., $\state \approx \sbasis_s \state_s$.
This way, the parameter-reduced system model
$A(\state, \pbasis_r \param_r) = 0$ can be approximated by the
following Galerkin projection:
\begin{equation}
  \underbrace{\sbasis_s^\top \lmodel(\pbasis_r^{} \param_r^{}) \sbasis_s^{}
  }_{\lmodel_s(\param_r)} \state_s^{} + \sbasis_s^\top  \nmodel(\sbasis_s^{} \state_s^{}, \pbasis_r^{} \param_r^{})  = 0.
  \label{eq:reduced_lf}
\end{equation}
If $s \ll m$, the dimension of the unknown state in
\eqref{eq:reduced_lf} is greatly reduced compared to that of the
original system \eqref{eq:amodel_LF}.
However, \eqref{eq:reduced_lf} cannot necessarily be solved quickly,
because the reduced-order model still requires evaluating the
full-scale system matrices or residual and then projecting those
matrices or the residual onto the reduced state subspace.
Many elements of these computations depend on the state and parameter
dimension of the original system, and hence this process is typically
computationally expensive (unless there is special structure to be
exploited, such as affine parametric dependence).
In this situation, methods such as missing point estimation
\cite{ROM:AWWB_2008}, empirical interpolation \cite{ROM:BMNP_2004}, or
its discrete variant \cite{ROM:ChaSor_2010}, can be used to
approximate the nonlinear term in the reduced-order model by selective
spatial sampling.

We employ the discrete empirical interpolation method (DEIM)
\cite{ROM:ChaSor_2010} to approximate the nonlinear terms of
\eqref{eq:reduced_lf}.
Suppose that the outputs of a nonlinear function
$\nmodel(\sbasis_s \state_s, \pbasis_r \param_r )$ in
\eqref{eq:reduced_lf} can be captured by a linear subspace spanned by
the basis $\fbasis_t \in \real^{m \times t}$. Then DEIM approximates
the nonlinear term by
\[
\nmodel(\sbasis_s \state_s, \pbasis_r \param_r ) \approx \fbasis_t
\alpha(\state_s, \param_r ).
\]
Here the lower-dimensional coefficient function
$\alpha(\cdot, \cdot) \in \real^t $ is constructed by selectively
evaluating the nonlinear function $\nmodel$ at output indices
$p_1, \ldots, p_t$, which are chosen by a greedy procedure.
This leads to a masking matrix
$\mask_t = [ \delta_{p_1}, \ldots, \delta_{p_{t}} ]$, where
$\delta_{i}$ is the canonical basis in $\real^m$ and the matrix
$\mask_t^\top \fbasis_t^{}$ is nonsingular.
This way, the coefficient function can be determined by
$\mask_t^\top \nmodel(\sbasis_s \state_s, \pbasis_r \param_r ) =
\mask_t^\top \fbasis_t^{} \alpha( \state_s,
\param_r ) $, which yields
\begin{equation}
  \alpha(\state_s^{}, \param_r^{} ) = \left( \mask_t^\top \fbasis_t^{} \right) ^{-1} \mask_t^\top
  \nmodel(\sbasis_s^{} \state_s^{}, \pbasis_r^{} \param_r^{}  ) .
\end{equation}
The resulting DEIM approximation of the reduced-order model
\eqref{eq:reduced_lf} becomes
\begin{equation}
  \lmodel_s^{}(\param_r^{}) \state_s^{} + \sbasis_s^\top  \fbasis_t^{}
  \alpha(\state_s^{}, \param_r^{} ) = 0,
  \label{eq:reduced_deim}
\end{equation}
and the associated model outputs are
\begin{equation}
  \label{eq:reduced_o}
  \mout = \omodel (\sbasis_s \state_s, \pbasis_r \param_r ).
\end{equation}
Together \eqref{eq:reduced_deim} and \eqref{eq:reduced_o} define a
reduced-order model $\mout = \forwardr (\param_r)$ that maps a
realization of the reduced parameter $\param_r$ to an approximation of
the observable model outputs.
In situations where the observables are high-dimensional or the
observation function $\omodel$ involves complex nonlinear relations,
we can again employ the DEIM method to approximate the observation
model.

\subsubsection{State subspace identification: choice of reference
  distribution}
\label{sec:stateref}

At the heart of model reduction is the identification of the reduced
bases $\sbasis_s$ and $\fbasis_t$.
We employ the well-known proper orthogonal decomposition (POD) method
\cite{ROM:Sirovich_1987, ROM:HLB_1996}, also known in statistics as
principal component analysis, for this task.
If the parameter $\param$ is distributed according to a probability
distribution $p$, the eigenvectors corresponding to the leading
eigenvalues of the state covariance matrix,\footnote{To be precise,
  \eqref{eq:cov_pod} is a matrix of second moments of the state, not
  the state covariance; realizations of the state are not
  centered. For simplicity, however, we still refer to this quantity
  and related expressions as state covariances.}
\begin{equation}
  K  = \int_{\paramspace} \state(\param) \state(\param)^\top p(d\param) ,
  \label{eq:cov_pod}
\end{equation}
form the reduced state basis $\sbasis_s$ in POD.
For dynamical problems, time variation of the state should also be
considered, in which case the right-hand side of \eqref{eq:cov_pod}
should be integrated over both time and the parameter.
The covariance matrix is often approximated empirically by full model
states---commonly referred to as snapshots
\cite{ROM:Sirovich_1987}---obtained at parameter samples drawn from
the probability distribution $p$.
Given $M$ samples and the snapshot matrix
$U = [ \state(\param^{(1)}), \ldots, \state(\param^{(M)}) ]$, the
empirical approximation of $K$ takes the form
\[
\widehat{K} = \frac{1}{M} U\,U^\top.
\]
If the sample size is much smaller than the state dimension, i.e.,
$M \ll m$, the singular value decomposition (SVD) of $U$ provides an
effective way of computing the reduced basis $\sbasis_s$.
The basis $\fbasis_s$ can be constructed in a similar fashion.

As in the parameter reduction problem, the choice of an appropriate
reference probability distribution for the parameters $\param$ is also
essential to identifying the reduced state subspace.
Using the construction of the basis $\sbasis_s$ as an example, we now
discuss various choices of reference distribution.
In the inverse problems literature, one common choice is the prior
distribution \cite{ROM:WangZab_2005,ROM:GFWG_2008,ROM:LSK_2013}, which
leads to the state covariance
\begin{equation}
  K_{\rm pr}  = \int_{\paramspace} \state(\param) \state(\param)^\top \prior(d\param) .
  \label{eq:cov_pod_prior}
\end{equation}
Alternatively, Cui et al.\ \cite{ROM:CMW_2014} suggest constructing
the reduced-order model over the support of the posterior rather than
the prior. The size and accuracy of the resulting reduced-order model
can scale better with parameter dimension than those of a
reduced-order model built from the prior, since the posterior has a
more concentrated support than the prior.
In the context of POD, this choice leads to the state covariance
\begin{equation}
  K_{\rm post}  = \int_{\paramspace} \state(\param) \state(\param)^\top \post(d\param | \data) .
  \label{eq:cov_pod_post}
\end{equation}

The choices above can have critical drawbacks for high-dimensional
ill-posed inverse problems, however.
Data may only inform a low-dimensional subspace in the parameter space
(the LIPS), within which the posterior distribution concentrates
relative to the prior.
In contrast, within the complement of the LIPS, the posterior and the
prior are essentially the same.
Hence the variation of the model states induced by either the prior or
the posterior is potentially dominated by the complement prior.
This effect is undesirable; recall that the goal of model reduction
here is to accelerate forward model simulations only for exploring the
parameter-reduced posterior. Thus, state variations induced by the
prior distribution on the complement of the LIPS should be eliminated.
This task can be readily achieved by choosing the parameter-reduced
posterior as the reference distribution, which leads to the state
covariance:
\begin{equation}
  \widetilde{K}_{\rm post}  = \int_{\paramspace_r} \state(\pbasis_r \param_r) \state(\pbasis_r \param_r)^\top \post(d\param_r \vert \data).
  \label{eq:cov_pod_rpost}
\end{equation}
As in the Posterior-LIPS construction, samples from the
parameter-reduced posterior used in a Monte Carlo approximation of
\eqref{eq:cov_pod_rpost} should also be adaptively selected during
posterior exploration.
The details of this adaptive sampling procedure will be given in
Section \ref{sec:algo}.

If the Laplace approximation \eqref{eq:lap} is used as an
approximation of the posterior, the effect of the complement prior can
be removed by projecting the full-dimensional input parameters onto
the LIPS.  This leads to the state covariance
\begin{equation}
  \widetilde{K}_{\rm L}  = \int_{\paramspace} \state(\projlis \param) \state(\projlis \param)^\top \lap(d\param),
  \label{eq:cov_pod_rlap}
\end{equation}
where $\projlis^{} = \pbasis_r^{} \Xi_r^\top$ is given in Definition
\ref{def:proj}.
The data-dependent nature of both the parameter-reduced posterior and
the Laplace approximation necessitate that model reduction be carried
out online when either is used as the reference distribution.

In contrast, using the prior projected onto the LIPS (referred to as
the parameter-reduced prior) as a reference distribution provides an
offline model reduction approach; here the state covariance becomes
\begin{equation}
  \widetilde{K}_{\rm pr}  = \int_{\paramspace_r} \state(\pbasis_r \param_r) \state(\pbasis_r \param_r)^\top \prior(d\param_r ).
  \label{eq:cov_pod_rprior}
\end{equation}
To render this approach fully offline, it should be applied together
with an offline parameter reduction method such as Prior-LIPS or
Prior-KL .

\begin{example}[State reduction with different reference
  distributions]
  \label{ex:liss}
  To demonstrate the effect of the reference distribution on state
  reduction, we consider the following linear example:
  \[
  u = \lmodel x, \quad {\rm and} \quad y = C u,
  \]
  where both the state and the parameters have dimension $n=m=200$,
  and $d=20$ observations are collected; thus we have
  $\lmodel \in \real^{200 \times 200}$ and
  $C \in \real^{20 \times 200}$.
  The system model $\lmodel$ and the prior covariance $\prcov$ have
  eigendecompositions of the form
  $\lmodel = \Psi_{\rm L}^{} \Delta_{\rm L}^{} \Psi_{\rm L}^\top$ and
  $\prcov = \Psi_{\rm pr}^{} \Delta_{\rm pr}^{} \Psi_{\rm pr}^\top $,
  and are constructed randomly.
  The orthonormal bases $\Psi_{\rm L}$ and $\Psi_{\rm pr}$ are
  computed by taking the QR decompositions of two independent square
  matrices with independent standard Gaussian entries. The spectra
  $\Delta_{\rm L} = \diag\{ \kappa_1, \ldots, \kappa_{200}\}$ and
  $\Delta_{\rm pr} = \diag\{ \rho_1, \ldots, \rho_{200}\}$ are
  prescribed as
  \[
  \kappa_i = \kappa_0 \left(\frac{i}{a_{\rm L}}\right)^{-b_{\rm L}},
  \quad {\rm and} \quad \rho_i = \rho_0 \left(\frac{i}{a_{\rm
        pr}}\right)^{-b_{\rm pr}}.
  \]
  We choose $\kappa_0 = 100$, $a_{\rm L} = 2$, $b_{\rm L} = 2$,
  $\rho_0 = 10$, $a_{\rm pr} = 10$ and $b_{\rm pr} = 4$ in this
  experiment.
  Observations are made at $d$ randomly selected indices of the state
  vector, and the observational noise is standard Gaussian.

  The left plot of Figure \ref{fig:mot_rom} shows the spectra of the
  prior covariance matrix $\prcov$ and of the system model $\lmodel$,
  along with the generalized eigenvalues \eqref{eq:eig_global}
  associated with the LIPS.
  The right plot of Figure \ref{fig:mot_rom} compares the spectra of
  the state covariance matrices induced by the prior, the posterior,
  the parameter-reduced prior, and the parameter-reduced posterior.
  We see that the spectra of the covariance matrices induced by the
  prior and the posterior decay more slowly than the spectra induced
  by the parameter-reduced prior and parameter-reduced posterior.
  This difference is due to the high-dimensional complement prior,
  which induces state variations that are irrelevant to the
  observations.
  Projection onto the LIPS eliminates the effect of the complement
  prior, and thus the corresponding state covariance matrices have
  quickly decaying spectra that in fact vanish at index $20$, which is
  the number of observations and the dimension of the LIPS.
  Note also that eigenvalues of the state covariance matrix induced by
  the parameter-reduced posterior decay more quickly than those
  induced by the parameter-reduced prior, since the posterior is more
  concentrated than the prior within the LIPS.
\end{example}

\begin{figure}[h!]
  \centerline{\includegraphics[width=0.65\textwidth]{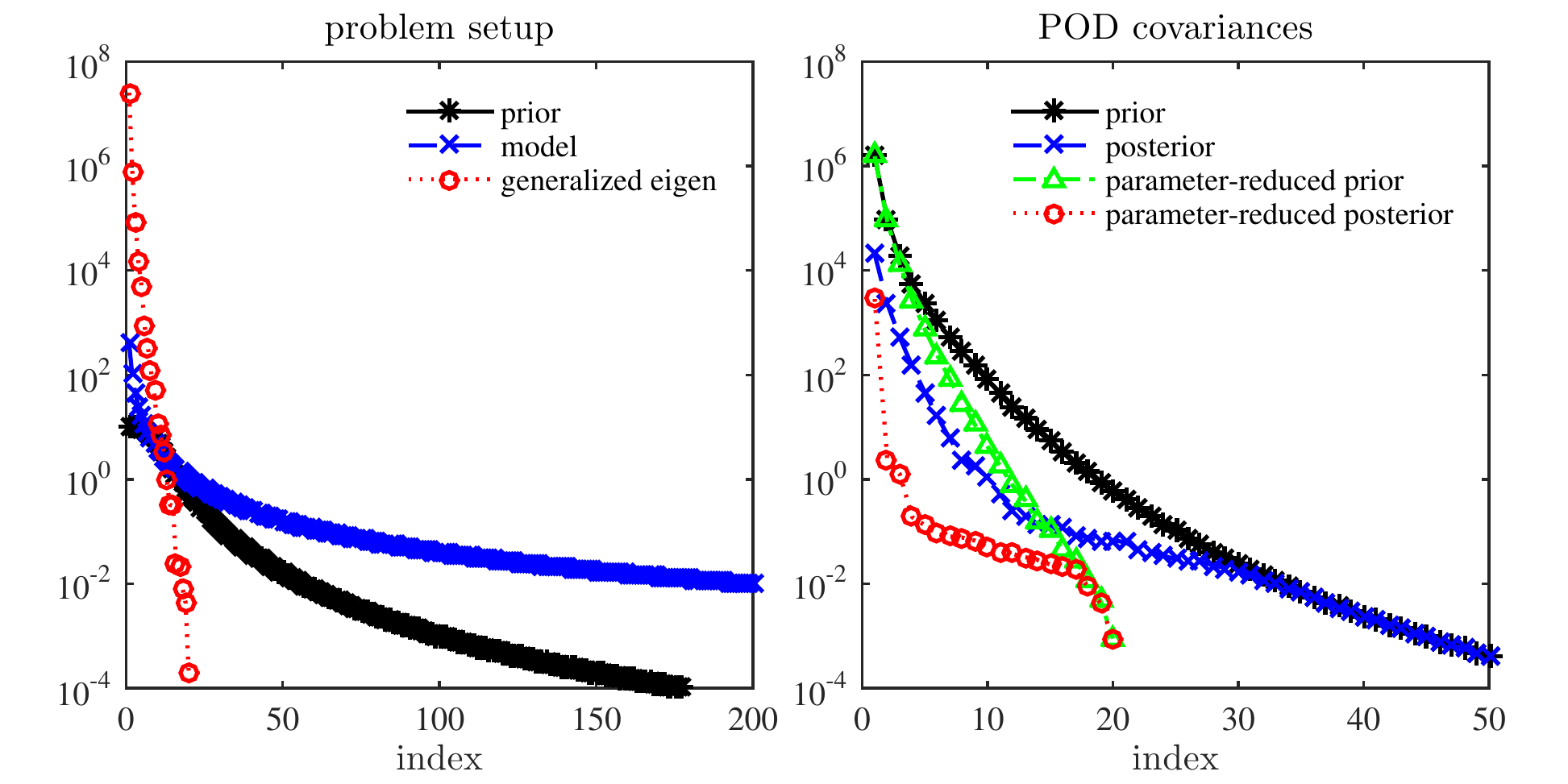}}
  \caption{The impact of parameter reference distributions on model
    reduction, as described in Example~\ref{ex:liss}.  Left: spectra
    of the prior covariance matrix $\prcov$ (stars) and the system
    model $\lmodel$ (crosses), along with generalized eigenvalues from
    \eqref{eq:eig_global} (circles). Right: spectra of the state
    covariance matrices induced by the prior, the posterior, the
    parameter-reduced prior, and the parameter-reduced posterior;
    these are denoted by stars, crosses, triangles and circles,
    respectively.}
  \label{fig:mot_rom}
\end{figure}

\newcommand{\kpost}{ \widetilde{K}_{\rm post} = \int_{\paramspace}
  \state(\pbasis_r \param_r) \state(\pbasis_r \param_r)^\top
  \post(d\param_r \vert \data)}
\newcommand{\kprior}{ \widetilde{K}_{\rm pr} \;\; = \int_{\paramspace}
  \state(\pbasis_r \param_r) \state(\pbasis_r \param_r)^\top
  \prior(d\param_r )}
\newcommand{\klap}{\widetilde{K}_{\rm L} \;\;\; = \int_{\paramspace}
  \state(\projlis \param) \state(\projlis \param)^\top \lap(d\param)}
\begin{table}[h!]
  \caption{Summary of state reduction methods: the associated state covariance matrix,
    whether the process is offline/online, and whether adaptive
    sampling is required.}
  \label{table:state_dim_redu}
  \begin{center}
    \begin{tabular}{l|l|l|l}
      \hline 
      Method & state covariance & online/offline & adaptive sampling? \\
      \hline 
      Posterior-LISS & $\kpost$  & online & yes\\
      Laplace-LISS & $\klap$ &  online & no\\
      Prior-POD & $\kprior$ & offline & no\\
      \hline
    \end{tabular}
  \end{center}
\end{table}

Since the amount of state variation in the reduced parameter subspace
is greatly reduced relative to that induced by either the
full-dimensional posterior or the full-dimensional prior, we will
henceforth construct reduced-order models only within the reduced
parameter subspace.
For each state reduction technique of this kind,
Table~\ref{table:state_dim_redu} summarizes the specific form of the
state covariance, whether reduction is offline/online, and whether
adaptive sampling is required, given the choice of parameter reference
distribution.
Overall, the benefit of pursuing model reduction on the reduced
parameter subspace is twofold: first, this approach addresses
challenges in the scalability of model reduction with parameter
dimension; second, it reduces the computational cost of handling
high-dimensional parameters in the parameterized reduced model.

\subsubsection{Jointly-approximated posterior}

By replacing the forward model $\forward: \real^n \rightarrow \real^d$
with the reduced-order model $\forwardr: \real^r \rightarrow \real^d$,
the data-misfit function \eqref{eq:misfit} can be approximated as
\begin{equation}
  \misfitr(\param_r) = \frac{1}{2}  \Vert
  \forwardr(\param_r) - \data   \Vert^2_{\obscov} .
  \label{eq:misfit_a}
\end{equation}
Then the resulting {\it jointly-reduced posterior} distribution has
the form
\begin{equation}
  \japost(\param_r|\data) \propto \exp \left( -  \misfitr(\param_r) \right ) \prior(\param_r).
  \label{eq:redu_post_j}
\end{equation}
Together with the complement prior, this defines a product-form {\it
  jointly-approximated posterior}
\begin{equation}
  \japost(\param \vert \data) \propto \japost(\param_r, \param_\perp \vert \data) \propto \japost(\param_r \vert \data) \prior(\param_\perp).
  \label{eq:approx_post_j}
\end{equation}
We use the descriptor `joint' to signify that both the parameter space
and the model state space have been reduced in these posterior
approximations. Because the high-dimensional complement prior
$\prior(\param_\perp)$ is analytically tractable and the
low-dimensional jointly-reduced posterior avoids computationally
expensive full forward model evaluations, this jointly-approximated
posterior can be used to design scalable and computationally efficient
posterior exploration schemes.

\begin{remark}[{\bf Data reduction}]
  In problems where the data set is high-dimensional, the computation
  time required to evaluate the data-misfit function \eqref{eq:misfit}
  can also be significant.
  These evaluations can be accelerated by exploiting low-dimensional
  structure in the data space, via the DEIM method.
  Suppose that the dimension of the output of the observation model,
  $\mout = \omodel (\sbasis_s \state_s, \pbasis_r \param_r )$, is much
  larger than the reduced parameter dimension and the reduced state
  dimension.
  Suppose also that the variation of the model outputs,
  $ \omodel (\sbasis_s \state_s, \pbasis_r \param_r )$, can be
  captured by a subspace spanned by a basis
  $\dbasis_o \in \real^{d \times o}$ that is orthogonal with respect
  to $\obsinner$.
  As in the model reduction case, the DEIM method can be used to
  identify a masking matrix
  $\mask_o = [ \delta_{p_1}, \ldots, \delta_{p_{o}} ]$, where
  $\delta_{i}$ is the canonical basis in $\real^o$, such that
  $\mask_o^\top \dbasis_o^{}$ is nonsingular.
  Thus we can determine a low-dimensional coefficient function
  \begin{equation}
    \beta(\state_s^{}, \param_r^{} ) = \left( \mask_o^\top \dbasis_o^{} \right) ^{-1} \mask_o^\top \,
    \omodel(\sbasis_s^{} \state_s^{}, \pbasis_r^{} \param_r^{}  ) ,
  \end{equation}
  which selectively evaluates the nonlinear function $\omodel$ at
  indices $p_1, \ldots, p_o$.
  The resulting approximated observation model has the form
  \begin{equation}
    \tilde{\mout} = \dbasis_o^{} \, \beta(\state_s^{}, \param_r^{} ) .
  \end{equation}
  In this setting, the coefficient function $\beta$ and the reduced
  system model \eqref{eq:reduced_deim} together define a new
  reduced-order forward model
  $\forwarddr : \real^{r} \rightarrow \real^{o}$ that has a smaller
  number of outputs.
  This corresponds to model outputs
  $\tilde{\mout} = \dbasis_o^{} \forwarddr (\param_r) $ in the
  original observable space.
  It leads to an approximated data-misfit function in the form of
  \begin{align}
    \misfitr(\param_r)  
    & = \frac{1}{2}  \left\|  \dbasis_o^{} \forwarddr (\param_r)  - \data \right\|^2_{ \obscov} \nonumber \\
    & =  \frac{1}{2}  \left\|  \forwarddr (\param_r)  - \dbasis_o^\top \obscov^{-1} \, \data \right\|^2 + c,
      \label{eq:reduced_data}
  \end{align}
  where
  $c = \frac{1}{2} \| ( I - \dbasis_o^{} \dbasis_o^\top) \, \data
  \|^2_{\obscov}$ is a constant.
  As in the state reduction case, we can also use a POD approach with
  the parameter-reduced posterior, the Laplace approximation, or the
  parameter-reduced prior as reference distributions for constructing
  the reduced data basis.
\end{remark}


\subsection{Integrated algorithms}
\label{sec:algo}

Concisely, construction of the jointly-approximated posterior
distribution involves parameter reduction followed by model reduction.
Combining the parameter reduction methods of Section \ref{sec:lips}
with the state reduction methods of Section \ref{sec:liss} leads to
various integrated strategies for this task.
We list several algorithmic options in Table \ref{table:joint_redu}, distinguished
according to their online/offline nature and whether adaptive sampling
is required.

\begin{table}[h]
  \caption{Summary of sampling strategies for constructing the jointly-approximated posterior.}
  \label{table:joint_redu}
  \begin{center}
    \begin{tabular}{l|l|l|l|l}
      \hline 
      Strategy  & \klpod &  \prdi & \lapdi & \podi  \\
      \hline 
      Parameter reduction & Prior-KL  & Prior-LIPS &  Laplace-LIPS & Posterior-LIPS \\
      State reduction   & Prior-POD &  Prior-POD & Laplace-LISS & Posterior-LISS\\
      Sampling requirement & offline & offline & online & online and adaptive \\
      \hline
    \end{tabular}
  \end{center}
\end{table}

\subsubsection{Non-iterative strategies}

The \klpod \ strategy is the simplest option above: we compute the reduced parameter
basis using the truncated KL expansion of the prior, and then compute
a POD basis for the state by sampling from the parameter-reduced prior.
In contrast, the \prdi \ strategy uses the prior expectation of the GNH
to construct a low-dimensional parameter subspace and defines the
parameter-reduced prior accordingly; this strategy is still entirely
offline, as both the parameter and state subspaces are independent of
the data.  Only prior samples are required in these two
strategies. Full model evaluations in the Prior-POD step and Hessian
evaluations in Prior-LIPS estimation can be massively parallelized.
The computational cost of \klpod \ is less than that of \prdi, since the
former involves no Hessians in the parameter reduction step.

Given observed data, the \lapdi \ strategy involves first finding the
posterior mode by solving an optimization problem. Then samples can be directly
drawn from the Laplace approximation \eqref{eq:lap}.
In this way, the GNH evaluations required for Laplace-LIPS parameter
reduction and the full model evaluations required for Laplace-LISS
state reduction can also be carried out in parallel.
The computational cost thus is comparable to that of the \prdi \ strategy.
However, in contrast with \prdi, \lapdi \ is an online strategy that
depends on a particular realization of the observed data---and on the
associated approximation of the posterior.
Using the data focuses attention on regions of high posterior
probability and leads to a localization effect in the identification
of subspaces; thus, we expect that the jointly-approximated posterior
computed by \lapdi \ will have better accuracy than those produced by
the \klpod \ or \prdi \ strategies, for comparable dimensions of the
parameter and state bases. We will explore this conjecture in
numerical results below.

\subsubsection{Iterative strategy}

In the \podi \ strategy, constructing parameter and state subspaces
requires computing expectations over the full posterior (to obtain the
posterior-LIPS) and the parameter-reduced posterior (to obtain the
posterior-LISS).
Since direct sampling is not feasible \textit{a priori}---after all,
we are computing these bases in order to \textit{facilitate} posterior
sampling---Algorithm \ref{algo:post_joint} proposes an iterative
sampling framework to construct the reduced bases adaptively during
posterior exploration.

\renewcommand{\algorithmicensure}{\textbf{Note:}}
\begin{algorithm}[h]
  \begin{algorithmic}[1]
    \Require{At iteration $k=0$, initialize the jointly-approximated
      posterior $\japost^{0}(\param | \data)$ to be either the prior
      or the Laplace approximation \eqref{eq:lap}. }
    \Require{At iteration $k$, given (i) the reduced data-misfit
      function $\misfitr^k(\param_r)$ and the jointly-reduced
      posterior $\japost^{k}(\param_r | \data)$ induced by the LIPS
      basis $\pbasis_r^{k}$ and the LISS basis $\sbasis_s^{k}$; (ii)
      the projector $\projlis^{k} = \pbasis_r^{k} (\Xi_r^k)^\top$ as
      in Definition \ref{def:proj}; and (iii) the resulting
      jointly-approximated posterior $\japost^{k}(\param | \data)$,
      one iteration of the algorithm is:}
    \If {$k = 0$}
    \State Generate two sample sets, $\{\param_i\}_{i = 1}^{N}$ and
    $\{\param_i\}_{i = 1}^{M}$, from $\japost^{0}(\param | \data)$ by
    direct sampling.
    \Else
    \State Generate two sample sets, $\{\param_i\}_{i = 1}^{N}$ and
    $\{\param_i\}_{i = 1}^{M}$, from $\japost^{k}(\param | \data)$ by
    applying MCMC to $\japost^{k}(\param_r | \data)$ and direct
    sampling to $\prior(\param_\perp)$.
    \EndIf
    \State Compute the LIPS basis $\pbasis_r^{k+1}$ by finding the
    dominant eigenvectors of
    $(\widehat{S}_{\rm post}^{}, \prcov^{-1}) $, where
    \begin{equation}
      \widehat{S}_{\rm post} = \frac{1}{\sum_{i = 1}^{N}
        \omega_i}\sum_{i = 1}^{N} \omega_i \, \hessian(\param_i) , \ 
      \ \text{and} \ \ \omega_i = \exp\left( \misfitr^k\left( (\Xi_r^k)^\top \param_i \right)  - \misfit(\param_i) \right).
      \label{eq:is_gnh}
    \end{equation}
    \State Update the projector
    $\projlis^{k+1} = \pbasis_r^{k+1} (\Xi_r^{k+1})^\top$ and compute
    the weighted snapshot matrix
    \begin{equation}
      U = \frac{1}{\sqrt{\sum_{i = 1}^{M} \upsilon_i}} \left[ \sqrt{\upsilon_1}
        \state\left(\projlis^{k+1} \param_1\right), \ldots, \sqrt{\upsilon_M}
        \state\left(\projlis^{k+1} \param_M \right) \right], 
      \label{eq:is_state}
    \end{equation}
    where
    $ \upsilon_i = \exp\left(\misfitr^k\left(
        (\Xi_r^{k})^\top \param_i \right) - \misfit\left(
        (\Xi_r^{k+1})^\top \param_i \right) \right)$,
    and then compute the reduced state basis $\sbasis_s^{k+1}$ via the
    SVD of $U$.
    \State For system models involving nonlinear functions, compute
    the DEIM basis $\fbasis_t^{k+1}$ as in Step 7, and then construct
    the masking matrix $\mask_t^{k+1}$.
    \State Update $\misfitr^{k+1}(\param_r)$,
    $\japost^{k+1}(x_r | \data)$ and $\japost^{k+1}(\param | \data)$.
    \Ensure{At $k = 0$, importance sampling is turned off, i.e.,
      $\omega_i = 1$ and $\upsilon_i = 1$.}
  \end{algorithmic}
  \caption{Iterative construction used in the \podi \ strategy.}
  \label{algo:post_joint}
\end{algorithm}

In Algorithm \ref{algo:post_joint}, we initialize the
jointly-approximated posterior to be either the prior or the Laplace
approximation \eqref{eq:lap}.
At each iteration $k$, two independent sets of samples are generated
from the current jointly-approximated posterior,
$\japost^{k}(\param | \data)$, in order to construct reduced parameter
and state bases using the Posterior-LIPS and Posterior-LISS methods,
respectively.
In this step, samples can be directly drawn when $k = 0$. For $k>0$,
sampling $\japost^{k}(\param | \data)$ requires applying MCMC to the
low-dimensional and cheap-to-evaluate jointly-reduced posterior
$\japost^{k}(\param_r | \data)$, and directly sampling the
high-dimensional complement prior $\prior(\param_\perp)$.
At each iteration $k$, we use the current jointly-approximated
posterior $\japost^{k}(\param | \data)$ as the biasing distribution to
compute the posterior expectation of the GNH via importance sampling;
this process can be written as
\begin{equation}
  S_{\rm post} = \int_\paramspace \frac{\post(\param |
    \data)}{\japost^{k}(\param | \data)} \, \hessian(\param) \,
  \japost^{k}(d\param | \data).
  \label{eq:is_gnh_1}
\end{equation}
Given the data-misfit function $\misfitr^k(\param_r)$ induced by the
current reduced order model, the importance weight
\begin{equation}
  \frac{\post(\param |
    \data)}{\japost^{k}(\param | \data)} \propto \omega(\param) := \exp\left( \misfitr^k(\param_r) - \misfit(\param) \right), \;\;  {\rm where} \;\; \param_r = (\Xi_r^k)^\top \param,
  \label{eq:is_gnh_2}
\end{equation}
can only be computed up to a normalizing constant.
This leads to the self-normalized importance sampling estimator of
$\widehat{S}_{\rm post}^{}$ in \eqref{eq:is_gnh}.
By finding the dominant eigenvectors of the pencil
$(\widehat{S}_{\rm post}^{}, \prcov^{-1}) $, we obtain the new LIPS
basis $\pbasis_r^{k+1}$.

\updated{
  \begin{remark}
    We note that it is not feasible to store and factorize the matrix
    $\widehat{S}_{\rm post}$ directly.
    By computing the action of the local GNH $\hessian(\param_i)$ on
    vectors---each action requires one forward model evaluation and
    one adjoint model evaluation---low-rank approximations of each
    sampled $\hessian(\param_i)$ can be computed using Krylov subspace
    methods \cite{Lin:GoVanlo_2012} or randomized algorithms
    \cite{Lin:HMT_2011, Lin:Liberty_etal_2007}.
    Monte Carlo estimates of the expected Hessians used to identify
    the Laplace-LIPS and the Prior-LIPS are constructed in the same
    way.
%
%
    We refer readers to \cite{DimRedu:Cui_etal_2014, MCMC:CLM_2014}
    for more details on storage management and computational strategies.
  \end{remark}}

The updated LIPS basis $\pbasis_r^{k+1}$ leads to a new
parameter-reduced posterior $\post^{k+1}(\param_r | \data)$, and then
the next task is to compute the new reduced-order model using the
Posterior-LISS.
Since finding the Posterior-LISS requires integration over
$\post^{k+1}(\param_r | \data)$, which can be computationally
expensive, we again employ importance sampling with the previous
jointly-approximated posterior $\japost^{k}(\param | \data)$ as the
biasing distribution.
We use the following identity
\newcommand{\subux}{\state\left(\pbasis_r^{k+1} \param_r^{}\right)}
\newcommand{\projux}{\state\left(\projlis^{k+1} \param\right) }
\begin{align*}
  \widetilde{K}_{\rm post}  
  & = \int_{\paramspace_r} \subux \subux ^\top \post^{k+1}(d\param_r^{} \vert \data), \nonumber \\
  & =  \int_{\paramspace} \projux \projux^\top \papost^{k+1}(d\param \vert \data),
\end{align*}
to derive the importance sampling formula in the full parameter space.
Thus, the state covariance estimated over
$\post^{k+1}(\param_r | \data)$ can be written as
\begin{equation}
  \widetilde{K}_{\rm post} = \int_{\paramspace}
  \frac{\papost^{k+1}(\param | \data)}{\japost^k(\param | \data)} \,
  \projux \projux^\top \, \japost^k(d\param | \data).
  \label{eq:is_state_1}
\end{equation}
As in the parameter reduction case, the likelihood ratio
\begin{equation}
  \frac{\papost^{k+1}(\param | \data)}{\japost^k(\param | \data)} \propto \upsilon(\param) := \exp\left( \misfitr^k\left( (\Xi_r^{k})^\top \param \right) - \misfit\left( (\Xi_r^{k+1})^\top \param \right)  \right),
  \label{eq:is_state_2}
\end{equation}
can only be computed up to a normalizing constant, and therefore we
use self-normalized importance sampling.
When the SVD is used to compute the POD basis, this leads to the
weighted snapshot matrix in \eqref{eq:is_state}.
We note that the full model evaluation in the parameter-reduced
data-misfit function $\misfit\left( (\Xi_r^{k+1})^\top \param \right)$
in \eqref{eq:is_state_2} generates exactly the snapshot
$\state(\projlis^{k+1} \param)$ used in computing the POD basis.

At the first iteration, the initial distribution
$\japost^{0}(\param | \data)$ can have a large discrepancy from the
posterior, and thus the resulting importance weights $\omega_i$ and
$\upsilon_i$ can potentially have large variances.
To overcome this potential sampling deficiency, we set the weights
$\omega_i$ and $\upsilon_i$ to $1$ at the first iteration.
This way, the initial distribution is used as a surrogate to explore
the support of the posterior, and importance sampling only kicks in at
later iterations to estimate the LIPS and the LISS.

\updated{
  \begin{remark}
    In each iteration of Algorithm \ref{algo:post_joint}, constructing
    the LIPS involves evaluating the forward model---and hence the
    full posterior density---at a set of samples, and computing the
    action of the GNH on a number of directions for each sample in this
    set.
    In addition, the forward model and the full posterior density are
    evaluated at another set of samples---projected onto the subspace
    spanned by the LIPS---to construct the reduced-order model.
    Based on our numerical experience, a sample size on the order of
    hundreds is sufficient for both steps.
    We also note that the first iteration of Algorithm
    \ref{algo:post_joint} generates a jointly-approximated posterior
    equivalent to the result of either the \emph{\prdi} strategy or
    the \emph{\lapdi} strategy, depending on the choice of initial
    distribution.
  \end{remark}}

To ensure the convergence of the self-normalized importance sampling
estimators \eqref{eq:is_gnh} and \eqref{eq:is_state}, it is required
that $\japost^{k}(\param | \data) > 0$ whenever
$\post(\param | \data) > 0$ , and that
$\japost^{k}(\param | \data) > 0$ whenever
$\papost^{k+1}(\param | \data) > 0$, for any $k>0$.
Furthermore, distributions constructed from the low-dimensional
subspaces are not guaranteed to capture the tails of the posterior
accurately, and hence the weights $\omega(\param)$ and
$\upsilon(\param)$ might have large variance in some situations.
Bounding $\omega(\param)$ and $\upsilon(\param)$ from above, however,
can control the variance of the importance weights and thus guarantee
finite variance of the estimators.
The following lemma establishes that by assigning upper bounds to the
approximated data-misfit functions, the ratios $\omega(\param)$ and
$\upsilon(\param)$ can be bounded.

\begin{lemma}
  \label{lemma:bounds}
  Given an upper bound $K > 0$ on the parameter-approximated
  data-misfit function $\misfit( \Xi_r^\top \param )$ and the
  jointly-approximated data-misfit function
  $\misfitr( \Xi_r^\top \param )$, i.e.,
  \begin{equation}
    \misfit( \Xi_r^\top \param ) \leq K < \infty \;\; {\rm and} \;\;
    \misfitr( \Xi_r^\top \param ) \leq K < \infty, 
    \label{eq:bounds}
  \end{equation}
  the ratios $\omega(\param)$ and $\upsilon(\param)$ defined in
  \eqref{eq:is_gnh_1} and \eqref{eq:is_state_2} are bounded as
  $\omega(\param) \leq \exp(K)$ and $\upsilon(\param) \leq \exp(K)$.
\end{lemma}
\begin{proof}
  Since all the data-misfit functions have the form of a weighted
  $L^2$ norm, their values cannot be negative, i.e.,
  $\misfit\left( \param \right) \geq 0$,
  $\misfit( \Xi_r^\top \param ) \geq 0$, and
  $\misfitr( \Xi_r^\top \param ) \geq 0$.
  The upper bounds on $\misfit( \Xi_r^\top \param )$ and
  $\misfitr( \Xi_r^\top \param )$ in \eqref{eq:bounds} then lead to
  \[
  \misfitr(\param_r) - \misfit(\param) \leq K \;\; {\rm and} \;\;
  \misfitr( \Xi_r^\top \param ) - \misfit( \Xi_r^\top \param ) \leq K.
  \]
  Thus both ratios $\omega(\param)$ and $\upsilon(\param)$ are bounded
  above by $\exp(K)$.
\end{proof}

We employ a heuristic based on a (somewhat frequentist) probabilistic
argument to choose a value of $K$ to impose as a bound on our misfit
functions. If the whitened residual in the data-misfit function,
$\obscov^{-1/2}(\forward(\param) - \data)$, is a $d$-dimensional
random vector whose components are independent standard Gaussians,
then the data-misfit function follows a chi-squared distribution with
$d$ degrees of freedom, $\chi_d^2$.
Then the upper bound $K$ can be chosen so that the probability of the
data-misfit function exceeding $K$ is $\tau_d \ll 1$, i.e.,
\[ {\mathbb{P}}[z > K] = \tau_d, \ \ \text{where} \ \ z \sim \chi_d^2.
\]
Here we choose $\tau_d = 10^{-4}$.

\subsubsection{Posterior exploration schemes}

The jointly-approximated posterior provides a launching point for many
scalable and computationally efficient posterior sampling schemes: the
analytically tractable and high-dimensional complement prior is
sampled directly, while the low-dimensional and analytically
intractable jointly-reduced posterior can be sampled by various MCMC
methods. Evaluations of the latter density are accelerated because
they rely only on the reduced-order forward model.
In this paper, we employ the adaptive Metropolis-adjusted Langevin
Algorithm (MALA) \cite{MCMC:Atchade_2006} to sample the jointly-reduced posterior.
We note that the separable representation of the jointly-approximated
posterior is amenable to a range of alternative posterior exploration
or integration approaches, e.g., implicit sampling
\cite{Sto:ChoTu_2009,Sto:MTAC_2012}, the randomize-then-optimize
method \cite{Sto:BSHL_2014}, and sparse quadrature
\cite{IP:SchStu_2012, IP:SchiSch_2013, IP:ChenSch_2015}.
If, on the other hand, one would like to compute the expectation of a
function of interest $g(\param)$ over the full posterior, samples from
the jointly-approximated posterior can be used to derive importance
sampling estimates thereof.
Bounding the jointly-approximated posterior as in Lemma
\ref{lemma:bounds}, and using the identity
\[
\expect \left[ g(\param) \right] = \int_\paramspace g(\param) \,
\post(d\param | \data) = \int_\paramspace \frac{\post(\param |
  \data)}{\japost(\param | \data)} \, g(\param) \, \japost(d\param |
\data),
\]
yields the estimator
\begin{equation}
  \expect \left[ g(\param) \right]  \approx \frac{1}{\sum_{i = 1}^{N}\omega(\param_i)}  \sum_{i = 1}^{N}  \omega(\param_i) g(\param_i),
  \label{eq:is_est}
\end{equation}
where
$ \omega(\param) = \exp\left( \misfitr\left( \Xi_r^\top \param \right)
  - \misfit\left( \param \right)\right) $
and $\param_i \sim \japost(\param | \data)$ for $i = 1,\ldots,N$.
This ratio $ \omega(\param)$ could also be used in a
delayed-acceptance MCMC method \cite{MCMC:ChriFox_2005, MCMC:Cui_2010}
to sample the full posterior; in this case, the jointly-approximated
posterior is used to ``screen'' MCMC proposals and thus more quickly
traverse the support of the full posterior.
Note that the importance sampling estimator above requires a full
posterior evaluation to compute each importance weight. This effort
can be computationally demanding, but these full posterior evaluations
can be massively parallelized, as the sampling step is based on the
jointly-approximated posterior.
Further variance reduction might be achieved using the control variates
technique \cite{Owen_2013}.
%



%
\section{Example 1: atmospheric remote sensing}
\label{sec:gomos}

In this section, we apply our joint approximation approach to a
realistic atmospheric remote sensing problem, where satellite
observations from the {\it Global Ozone MOnitoring System} (GOMOS) are
used to estimate the concentration profiles of various gases in the
atmosphere.
We will first present the GOMOS model, the inverse problem setup, and
the reduced-order model.
Then we will demonstrate various aspects of the joint posterior
approximation using the GOMOS inversion.

\subsection{Problem setup}

The GOMOS instrument repeatedly measures light intensities $\rho_\nu$
at different wavelengths $\nu$. First, a reference intensity spectrum
$\rho_{\rm ref}$ is measured above the atmosphere. The transmission
spectrum is defined as $T_\nu= \rho_\nu / \rho_{\rm ref}$. The
transmissions measured at wavelength $\nu$ along the ray path $z$ are
modelled using Beer's law:
\begin{equation}
  T_{\nu,z}=\exp \left( - \int_z \sum_{\mathrm{gas}} a_\nu^{\mathrm{gas}}(z(\zeta)) \kappa^{\mathrm{gas}}(z(\zeta))d\zeta \right),
  \label{mod}
\end{equation}
where $\kappa^{\mathrm{gas}}(z(\zeta))$ is the density of a gas (unknown
parameter) at tangential height $z$. The so-called cross-sections
$a_\nu^{\mathrm{gas}}$, known from laboratory measurements, define how
much a gas absorbs light at a given wavelength.

To approximate the integrals in (\ref{mod}), the atmosphere is
discretized. The geometry used for inversion resembles an onion: the
gas densities are assumed to be constant within spherical layers
around the Earth. The GOMOS measurement principle is illustrated in
Figure \ref{fig:gomos_setting}.

\begin{figure}[h!]
  \centerline{\includegraphics[width=0.8\textwidth]{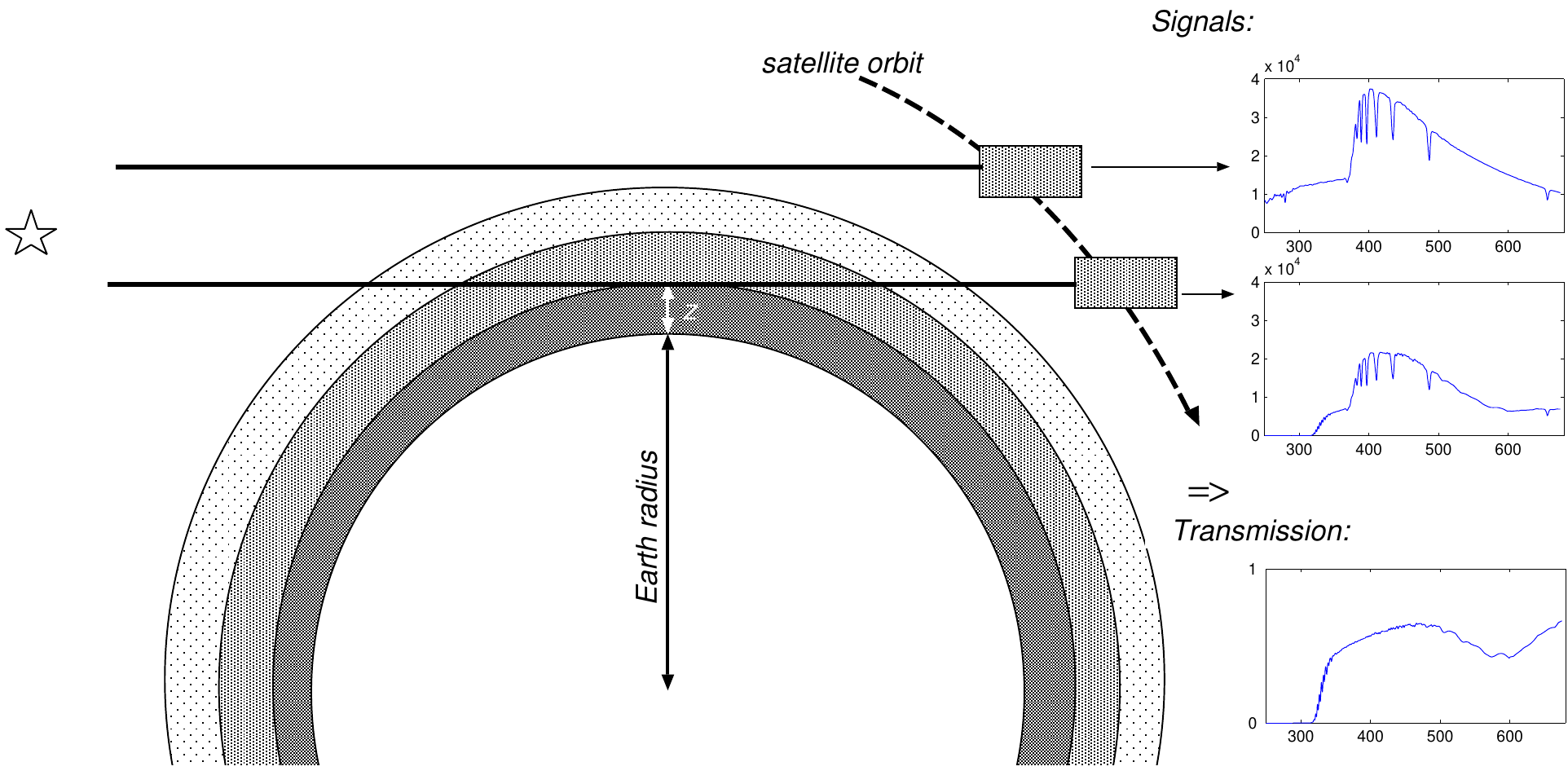}}
  \caption{The principle of the GOMOS measurement. The reference
    intensity is measured above the atmosphere. The observed
    transmission spectrum is the attenuated spectrum (measured through
    the atmosphere) divided by the reference spectrum. The atmosphere
    is presented locally as spherical layers around the Earth. Note
    that the thickness of the layers is much larger relative to the
    Earth in this figure than in reality. The figure is adopted from
    \cite{IP:Haario_etal_2004}, with the permission of the authors.}
  \label{fig:gomos_setting}
\end{figure}

We assume that the cross-sections do not depend on height. In the
inverse problem we have $N_{\rm gas}$ gases, $N_{\nu}$ wavelengths,
and the atmosphere is divided into $N_{\rm alts}$ layers. The
discretization is fixed so that number of measurement lines is equal
to the number of layers. Approximating the integrals by sums over the
chosen grid, and combining information from all lines and all
wavelengths, we can write the model in matrix form as follows:
\[
T=\exp(-A B^\top G^\top),
\]
where $T \in \real^{N_{\nu} \times N_{\rm alts}}$ are the modelled
transmissions, $A \in \real^{N_{\nu} \times N_{\rm gas}}$ contains the
cross-sections, $B \in \real^{N_{\rm alts} \times N_{\rm gas}}$
contains the unknown densities, and
$G \in \real^{N_{\rm alts } \times N_{\rm alts}}$ is the geometry
matrix that contains the lengths of the lines of sight in each
layer. %

Each gas density profile is endowed with an independent log-normal
process prior.
Equivalently, the gas density profiles can be represented as
$B = \exp(X)$.
Using MATLAB notation, each discretized log-profile $X(:,i)$ follows a
Gaussian prior $ N(\prmean^{(i)},\prcov^{(i)})$, where $\prcov^{(i)}$
is defined by the squared exponential covariance kernel
\begin{equation}
C_i(\zeta, \zeta^\prime)=\sigma_i\exp \left(-\frac{\|\zeta - \zeta^\prime\|^2}{2\zeta_{0}^2} \right),
\end{equation}
where the correlation length is $\zeta_{0}=10$. In the example below,
we will infer $N_{\rm gas}= 4$ unknown gas profiles; thus we choose
$\sigma_1=5.22$, $\sigma_2=9.79$, $\sigma_3=23.66$, and
$\sigma_4=83.18$.
These priors are chosen to promote smooth gas density profiles with large variations.

Now let $\otimes$ denote the Kronecker product and $\mathrm{vec}(\cdot)$
be a vectorization operator that stacks the columns of its matrix argument on
top of each other.
Using the identity $\mathrm{vec}(AB^\top G^\top)=(G \otimes A) \mathrm{vec}(B^\top)$, we
obtain the vectorized parameter-to-data relationship,
\begin{equation}
  \data = \mathrm{vec}(T)+\error = \exp \left( -(G \otimes A) \exp( \param ) \right)+\error,
  \label{eq:modvec}
\end{equation}
where $\param = \mathrm{vec}(X^\top) $ are the vectorized parameters
and $\error$ is the measurement error modeled by independent Gaussian
random variables with known variances.
Here we adopt the same model setup and synthetic data set used in
\cite{DimRedu:Cui_etal_2014}.
The atmosphere is discretized into $N_{\rm alts} = 50$ layers, and
with four profiles to infer, the total dimension of the parameter
is $n = 200$.
We have observations at $N_\nu = 1416$ wavelengths, and thus the
dimension of the data is $d = 70800$.
\updated{Although the data dimension is much higher than the parameter
  dimension in this case, the resulting inverse problem is still
  ill-posed. The forward model introduces a strong smoothing effect, and
  thus the high-dimensional data can inform only a small number of
  parameter dimensions. Similar situations are encountered in X-ray
  tomography \cite{DimRedu:Spantini_etal_2015}, where the forward
  model also involves a system of integral equations.}
We refer the readers to \cite{DimRedu:Cui_etal_2014} for a further
description of the model setup and the data set.
For more details about the GOMOS instrument and the Bayesian treatment
of the inverse problem, see \cite{IP:Haario_etal_2004,MCMC:Tamminen_2004} and the
references therein.

The forward model
$\mout = \exp \left( -(G \otimes A) \exp( \param ) \right)$ maps from
$\real^{200}$ to $\real^{70800}$ and involves two exponential
functions.
Given a reduced parameter basis $\pbasis_r$ and a realization of the
reduced parameter $\param_r$, the computational expense of evaluating
the forward model with the reduced parameter,
$\exp \left( -(G \otimes A) \exp( \pbasis_r \param_r ) \right)$, arises
from several sources: the exponential expression
$\exp( \pbasis_r \param_r ) $, which involves a matrix-vector product
and a $200$-dimensional exponential function evaluation; the
matrix-vector product with $G \otimes A$; and the evaluation of the
$70800$ dimensional exponential function for producing the model
outputs.
To set up the reduced-order model, the first task is to construct a DEIM
interpolation for the exponential function
$\exp( \pbasis_r \param_r ) $.
Given a basis $\fbasis_t$ that spans the subspace capturing the
variations of the outputs of $\exp( \pbasis_r \param_r )$ and the
associated masking matrix $\mask_t$, this DEIM interpolation takes the form
\[
\exp( \pbasis_r \param_r ) \approx \fbasis_t^{} \alpha(\param_r^{}),
\quad {\rm where} \quad \alpha(\param_r^{}) = (\mask_t^\top
\fbasis_t^{} )^{-1} \exp( \mask_t^\top \pbasis_r^{} \param_r^{} ) .
\]
Then, given a reduced data basis $\dbasis_o$ and the associated
masking matrix $\mask_o$, another DEIM interpolation is employed to
reduce the \textit{output} dimension of the forward model in the form
of
\begin{equation}
  \forwardr (\param_r^{}) = (\mask_o^\top \dbasis_o^{})^{-1} \exp \left( - \mask_o^\top (G \otimes A)  \fbasis_t^{}  \alpha( \param_r^{} ) \right).
  \label{eq:modvec_r}
\end{equation}
In the reduced-order model above, the computational cost is dominated
by the evaluation of the nonlinear function
$\alpha: \real^{r} \rightarrow \real^{t}$, the matrix-vector product
with the $o \times t$ dimensional matrix
$ \mask_o^\top (G \otimes A) \fbasis_t^{} $, and the $o$-dimensional
exponential function.
The reduced-order model \eqref{eq:modvec_r} is used together with the
approximated data-misfit function \eqref{eq:reduced_data} to
accelerate evaluations of the original data-misfit function, which
involved high-dimensional model outputs.

\subsection{Numerical results}

We first benchmark the parameter reduction methods introduced in
Section \ref{sec:lips}.
The (squared) Hellinger distance\footnote{The Hellinger distance
  translates directly into bounds on expectations
  \cite{IP:Stuart_2010}, and hence we use it as a metric to quantify
  the error of approximated posterior distributions.} between the full
posterior $\post(\param \vert \data) $ and the parameter-approximated
posterior $\papost(\param \vert \data)$,
\begin{equation}
  D_{\rm H}^2 \left( \post(\param), \papost(\param) \right) = \frac12 \int_\paramspace \left( \sqrt{\post(\param \vert \data)} - \sqrt{\papost(\param \vert \data)} \right)^2 d\param,
  \label{eq:hell_p}
\end{equation}
is used to evaluate the errors induced by various parameter reduction
methods, and to examine convergence versus the number of basis vectors
used for the parameter subspace (LIPS or Prior-KL).
Results are shown in Figure \ref{fig:gomos_param_dh}.
In this example, all three likelihood-informed methods converge more
quickly than Prior-KL (triangles).
As expected, Prior-LIPS (squares) is outperformed by the other two
likelihood-informed methods, and Posterior-LIPS (crosses) is more
accurate than the other methods for any given parameter subspace
dimension.
We note that the Laplace approximation itself (dashed line) has a
rather large Hellinger distance from the posterior. This reflects the
non-Gaussianity of the problem, and should not be confused with the
fact that the convergence curve of the non-Gaussian Laplace-LIPS
approximation (diamonds) is sandwiched between those of Prior-LIPS and
Posterior-LIPS.

\begin{figure}[h!]
  \centerline{\includegraphics[width=0.7\textwidth]{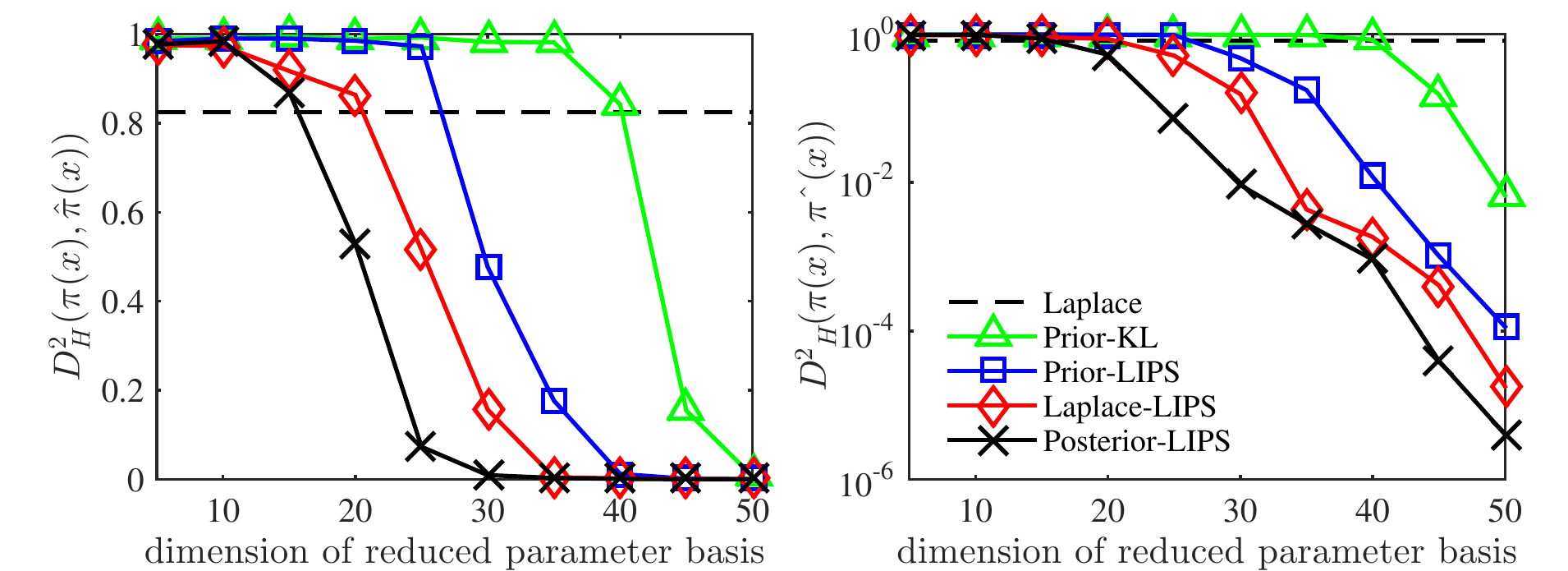}}
  \caption{Convergence of various \textit{parameter reduction}
    methods, for the GOMOS example: squared Hellinger distances
    \eqref{eq:hell_p} versus the dimensions of the reduced parameter
    bases $\pbasis_r$ defined by Posterior-LIPS, Laplace-LIPS,
    Prior-LIPS, and Prior-KL. The dashed line represents the squared
    Hellinger distance between the full posterior and its Laplace
    approximation \eqref{eq:lap}. Both subfigures show the same data,
    but for clarity we plot with both linear (right) and semilog
    (left) scales.}
  \label{fig:gomos_param_dh}
\end{figure}

Given a $28$-dimensional parameter basis built by Posterior-LIPS, we
next benchmark the state reduction methods introduced in Section
\ref{sec:liss}.
The (squared) Hellinger distance between the parameter-reduced
posterior $\post(\param_r \vert \data)$ and the jointly-reduced
posterior $\japost(\param_r \vert \data)$,
\begin{equation}
  D_{\rm H}^2 \left( \post(\param_r), \japost(\param_r) \right) = \frac12 \int_{\paramspace_r} \left( \sqrt{\post(\param_r \vert \data)} - \sqrt{\japost(\param_r \vert \data)} \right)^2 d\param_r,
  \label{eq:hell_m}
\end{equation}
is used to compare the convergence of various state reduction methods,
as a function of the number data basis vectors used in the
approximated data-misfit function \eqref{eq:reduced_data}.
Four DEIM approximations of the exponential function
$\exp( \pbasis_r \param_r )$, with reduced bases $\Theta_t$ of
dimension $20$, $40$, $60$ and $80$, are used in this benchmark.
Results are shown in Figure \ref{fig:gomos_rom_dh}.
In this test, Prior-POD fails to produce a jointly-reduced posterior
of reasonable accuracy (we observe
$D_{\rm H}^2 \left( \post(\param_x), \japost(\param_r) \right)$ always
above 0.9), and thus its performance is not reported.
For both Posterior-LISS and Laplace-LISS, the convergence of the
resulting jointly-reduced posteriors depends on both the DEIM
interpolation of the exponential function $\exp( \pbasis_r \param_r )$
and on the dimension of the reduced data basis. If the first DEIM
interpolation is too coarse (e.g., $\mydim(\fbasis_t)= 20$ or $40$),
the error that can be achieved by refining the reduced data basis
reaches a plateau.
But for any DEIM interpolation, the jointly-reduced posterior induced
by Posterior-LISS is about two orders of magnitude more accurate than
that produced by Laplace-LISS.

\begin{figure}[h!]
  \centerline{\includegraphics[width=0.7\textwidth]{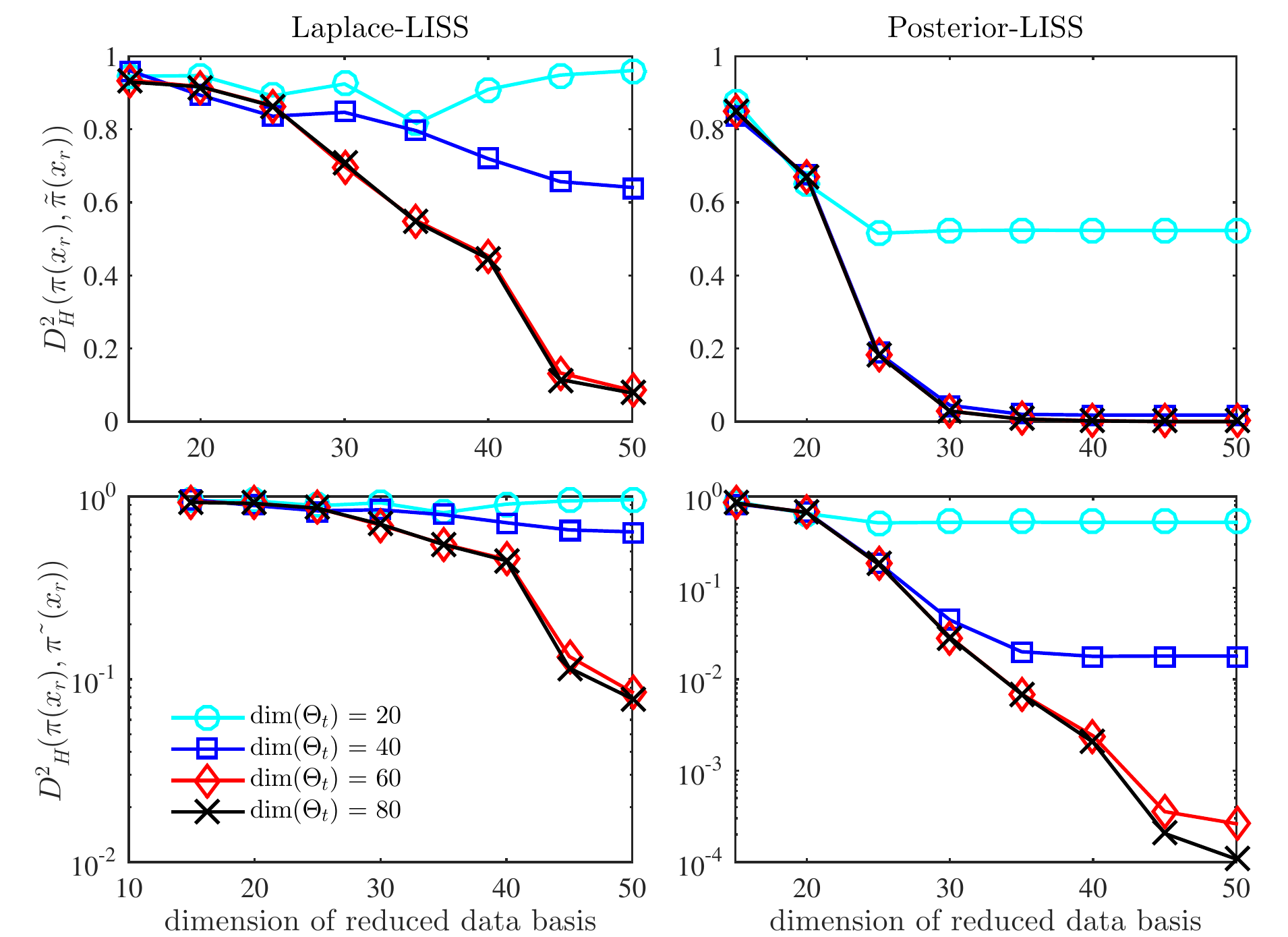}}
  \caption{Convergence of various \textit{data reduction} methods, for
    the GOMOS example. Squared Hellinger distances \eqref{eq:hell_m}
    versus the dimension of reduced data basis $\dbasis_o$ defined by
    Laplace-LISS and Posterior-LISS are shown in the left and right
    columns, respectively.  Four different DEIM approximations of the
    exponential function $\exp( \pbasis_r \param_r )$, with dimensions
    $20$, $40$, $60$, and $80$, are shown for each case. Again, errors
    are plotted on linear (top) and logarithmic (bottom) scales.}
  \label{fig:gomos_rom_dh}
\end{figure}

In the previous comparisons, the posterior-oriented methods
(Posterior-LIPS and Posterior-LISS) show clear advantages over the
other methods.
Thus, we would expect the various strategies for constructing the
jointly-approximated posterior in Section \ref{sec:algo} to have
similar performance characteristics.
We now demonstrate eight iterations of the \podi \ strategy (Algorithm
\ref{algo:post_joint}) using either the prior or the Laplace
approximation \eqref{eq:lap} as initial distributions.
The reduced parameter bases are truncated at the eigenvalue threshold
threshold $\tau_{\rm g} = 0.1$, the DEIM basis $\fbasis_t$ for
interpolating the function $\exp( \pbasis_r \param_r )$ is truncated
to retain eigenvalues above $10^{-12}$, and the reduced data basis is
truncated to retain eigenvalues above $10^{-5}$.
\updated{To build the LIPS in each iteration of \podi, the forward
  model and the action of the GNH in multiple directions are evaluated
  at 200 parameter samples. In particular, for each parameter sample,
  we use one forward model simulation and the action of GNH on 30
  directions. To build the reduced-order model at each iteration, we
  evaluate the forward model at 500 parameter samples.
  In comparison, the number of forward model and GNH--action
  evaluations required by \lapdi \ or \prdi \ is exactly the same as
  that required by one iteration of \podi.

  To find the MAP, we employ the subspace trust region method of
  \cite{Opt:CoLi_1994, Opt:CoLi_1996} with inexact Newton
  iterations. Using the prior mean as the initial guess, finding the MAP
  requires 25 Newton iterations. Each Newton iteration involves one
  forward model evaluation and the action of the GNH on an average of 8
  directions.
  The computational cost of finding the MAP is much smaller than that
  of constructing the jointly-approximated posteriors.
  But we emphasize that sampling the jointly-approximated posteriors
  does not involve any further full forward simulations.  }

\begin{figure}[h!]
  \centerline{\includegraphics[width=0.7\textwidth]{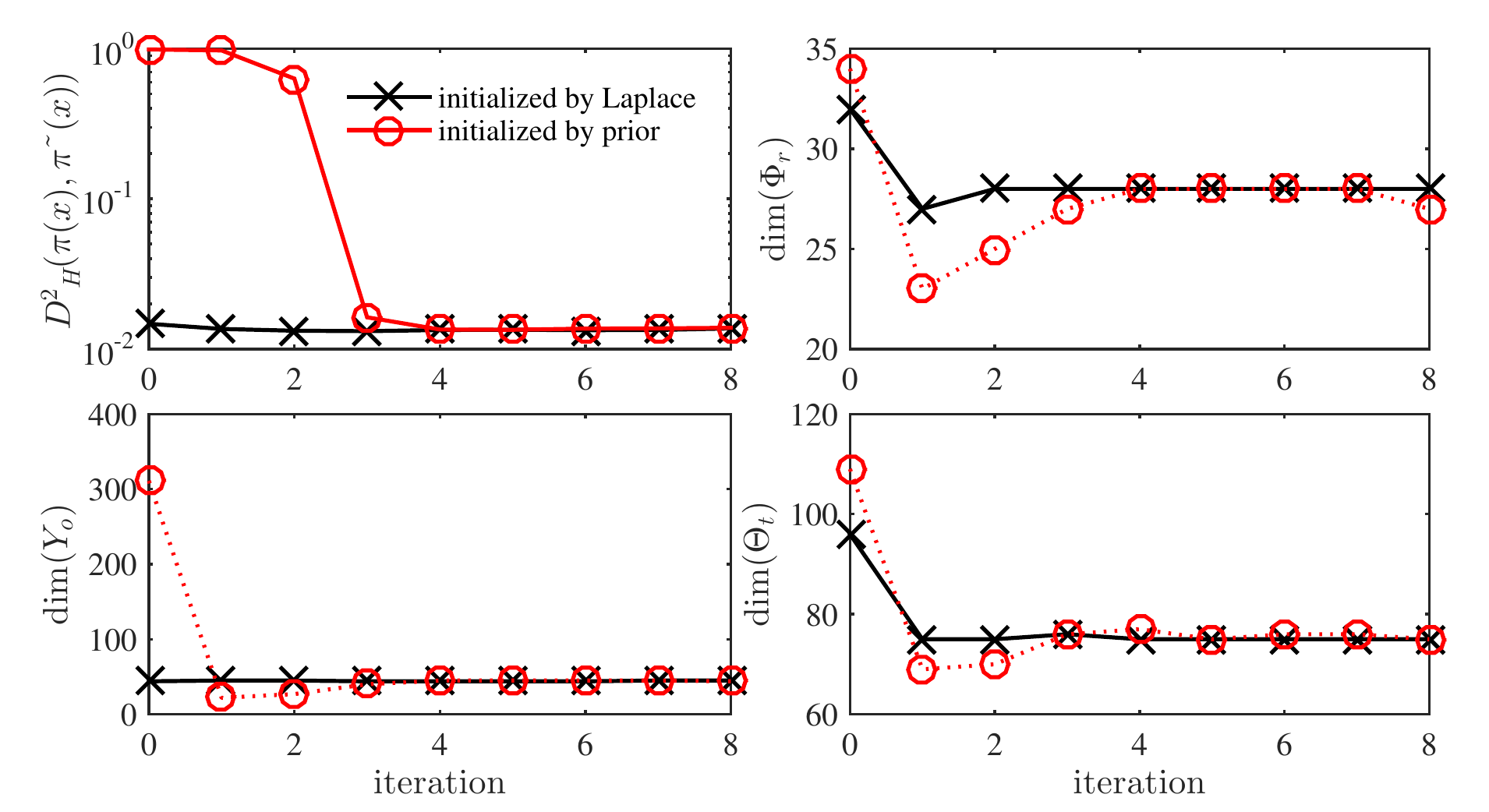}}
  \caption{Iterations of Algorithm \ref{algo:post_joint}, for the
    GOMOS example. The logarithm of the squared Hellinger distance
    \eqref{eq:hell_j}, and the dimensions of the reduced bases
    $\pbasis_r$, $\dbasis_o$, and $\fbasis_t$ are shown. We show
    iteration histories beginning with two different initial
    distributions: the prior and the Laplace approximation
    \eqref{eq:lap}.}
  \label{fig:gomos_joint}
\end{figure}

As in the previous comparisons, we use the (squared) Hellinger
distance between the full posterior $\post(\param \vert \data)$ and
the jointly-approximated posterior $\japost(\param \vert \data)$,
\begin{equation}
  D_{\rm H}^2 \left( \post(\param), \japost(\param) \right) = \frac12 \int_{\paramspace} \left( \sqrt{\post(\param \vert \data)} - \sqrt{\japost(\param \vert \data)} \right)^2 d\param,
  \label{eq:hell_j}
\end{equation}
as the error measure.
We also compare the dimension of reduced parameter basis $\pbasis_r$,
the dimension of the DEIM basis $\fbasis_t$, and the dimension of the
reduced data basis $\dbasis_o$.
The results are shown in Figure \ref{fig:gomos_joint}.
In this example, when the Laplace approximation is used as the initial
distribution, the Hellinger distance \eqref{eq:hell_j} remains flat
for all iterations, and the dimensions of the various reduced bases
stabilize in the first iteration.
In contrast, when the prior distribution is used as the initial
distribution, the algorithm stabilizes at iteration 4, and the
Hellinger distance \eqref{eq:hell_j} is rather large in the first
three iterations ($> 0.9$).
We recall that the first iteration of Algorithm \ref{algo:post_joint}
generates a jointly-approximated posterior corresponding to either
\prdi \ or \lapdi, depending on the initial distribution.
This difference in initial errors also shows that the \lapdi \
strategy can (by itself) be useful for constructing a
jointly-approximate posterior in this case, while the \prdi \ strategy
is not able to provide an accurate approximation.

In Figure \ref{fig:gomos_marginal}, we also plot the full posterior
and the jointly-approximated posteriors generated by \podi \ (at
iteration 8), \lapdi, and \prdi, as well as the Laplace approximation
\eqref{eq:lap}, marginalized onto the first eight KL basis functions.
\updated{
Here the full posterior is sampled by the DILI MCMC algorithm of
\cite{MCMC:CLM_2014}, which is an exact sampling method.}
Both \lapdi \ and \podi \ yield marginal distributions that are almost
identical to those of the full posterior, whereas the marginals of
\prdi \ demonstrate large discrepancies with the full posterior.
Overall, for this example, running \lapdi \ is the most
computationally efficient way to generate the jointly-approximated
posterior. By running an additional iteration of \podi, the dimensions
of the reduced bases can be further reduced without loss of accuracy.

\begin{figure}[h!]
  \centerline{\includegraphics[width=\textwidth]{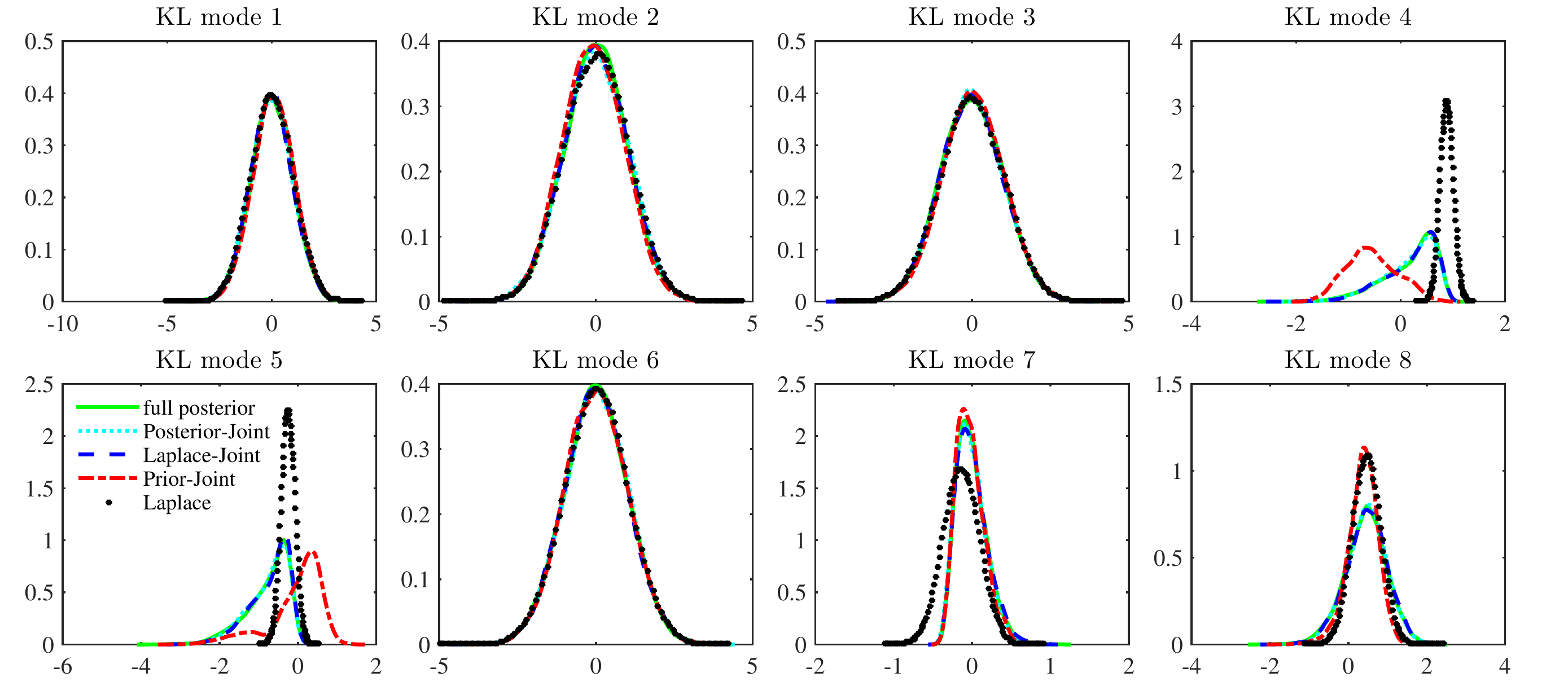}}
  \caption{Marginal posterior distributions for the GOMOS example: the
    full posterior and the jointly-approximated posteriors generated
    by \podi, \lapdi, and \prdi, along with the Laplace approximation
    \eqref{eq:lap}. We show marginal distributions along each of the
    first eight KL modes.}
  \label{fig:gomos_marginal}
\end{figure}


\section{Example 2: groundwater aquifer inversion}
\label{sec:elliptic}

Our second example is an elliptic PDE coefficient inverse problem. In
physical terms, our problem setup corresponds to inferring the
transmissivity field of a two-dimensional groundwater aquifer from
partial observations of the stationary drawdown field of the
watertable, measured from well bores.

\subsection{Problem setup}
\label{sec:ellpitic_setup}
Consider a three kilometer by one kilometer problem domain
$\Omega = [0\,\meter, 3000\,\meter]\times [0\,\meter, 1000\,\meter]$,
with boundary $\partial \Omega$. We denote the spatial coordinate by
$\zeta \in \Omega$.
Consider the transmissivity field $T(\zeta)$ (units
[$\meter^2/\unittime$]), the drawdown field $\state(\zeta)$ (units
[$\meter$]), and sink/source terms $q(\zeta)$ (units
[$\meter/\unittime$]).
The drawdown field for a given transitivity and source/sink
configuration is governed by the elliptic equation:
\begin{equation}
  -\nabla \cdot \left( T(\zeta) \nabla \state(\zeta) \right) = q(\zeta), \quad \zeta \in \Omega .
  \label{eq:forward_e}
\end{equation}
We prescribe the drawdown value to be zero on the boundary (i.e., a
Dirichlet boundary condition), and define the source/sink term
$q(\zeta)$ as the superposition of four weighted Gaussian plumes with
standard width $50$ meters. The plumes are centered near the four
corners of the domain (at $[20\,\meter, 20\,\meter]$,
$[2980\,\meter, 20\,\meter]$, $[2980\,\meter, 980\,\meter]$ and
$[20\,\meter, 980\,\meter]$) with magnitudes of --3000, 2000, 4000,
and --300 [$\meter/\unittime$], respectively.
We solve \eqref{eq:forward_e} by a finite element method.

The discretized transitivity field $T(\zeta)$ is endowed with a
log-normal prior distribution, i.e.,
\begin{equation}
  \label{eq:prior_e}
  T = \exp(\param), \; {\rm and} \; \param \sim \normal\left(\prmean, \prcov \right),
\end{equation}
where the prior mean is set to $\log( 1000\,[\meter/\unittime] )$ and
the inverse of the covariance matrix $\prcov^{-1}$ is defined through
the discretization of an Laplace-like stochastic partial differential
equation \cite{MRF:LRL_2011},
\begin{equation}
  \label{eq:corr}
  ( -\nabla \cdot K \nabla + \kappa^2 ) \param(\zeta) = \mathcal{W}(\zeta),
\end{equation}
where $\mathcal{W}(\zeta)$ is white noise.
As discussed in \cite{IP:Stuart_2010}, this way of defining the
precision operator of a Gaussian prior is discretization invariant,
i.e., the posterior distribution will converge to its functional limit
under grid refinement.
In this example, we set the stationary, anisotropic correlation tensor
$K$ to
\[
K = \left[\begin{array}{rr} 0.55 & -0.45\\ -0.45 &
    0.55 \end{array}\right],
\]
and put $\kappa = 50$.
The ``true'' transmissivity field is a realization from the prior
distribution. The true transmissivity field, the sources/sinks, the
simulated drawdown field, and the synthetic data are shown in Figure
\ref{fig:setup_e}.
Partial observations of the pressure field are collected at $d=13$
sensors whose locations are depicted by black dots in Figure
\ref{fig:setup_e}(c).  The observation operator $\omodel$ is simply
the corresponding ``mask'' operation.  This yields observed data
$\data \in \real^{13}$ as
\[
\data = \omodel \state(\zeta) + \error ,
\]
with additive error $\error \sim \normal(0, \sigma^2 I_{13})$.
\updated{The standard deviation $\sigma$ of the measurement noise is prescribed
so that the observations have signal-to-noise ratio $120$, where the
signal-to-noise ratio is defined as
$\mathbb{V}{\rm ar}(\data)/\sigma^2$.}
The noisy data are shown in Figure \ref{fig:setup_e}(d).

\begin{figure}[h!]
  \centerline{\includegraphics[width=0.8\textwidth]{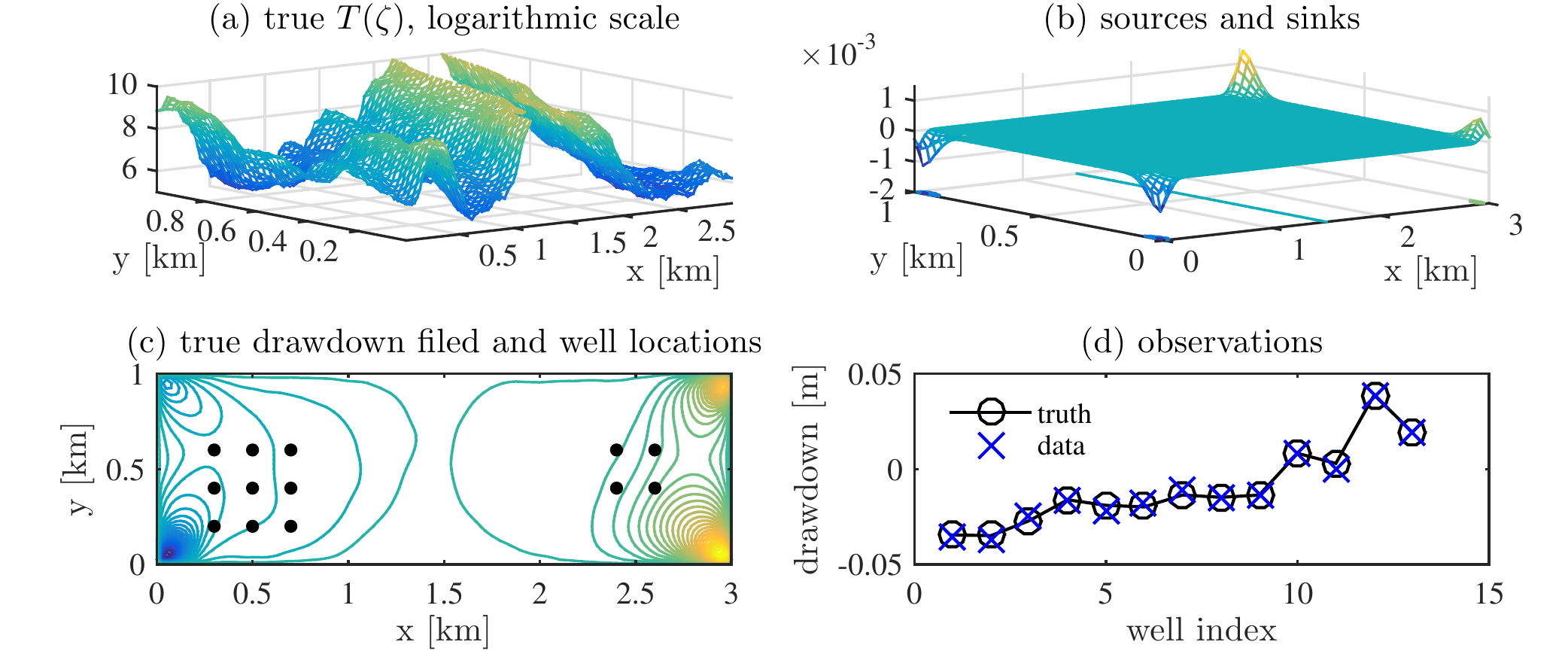}}
  \caption{Setup of the groundwater inversion example. (a) ``True''
    transmissivity field. (b) Sources and sinks. (c) Drawdown field
    resulting from the true transmissivity field, with observation
    wells indicated by black dots. (d) Data $\data$; circles represent
    the noise-free drawdowns at each well, while crosses represent the
    observed drawdowns corrupted with measurement noise.}
  \label{fig:setup_e}
\end{figure}

In this example, the finite element discretization of
\eqref{eq:forward_e} uses $120\times 40$ bilinear elements to
represent the drawndown field $\state(\zeta)$, while the
transmissivity field $T(\zeta)$ is modeled as piecewise constant for
each element. This yields the discretized system of equations
\begin{equation}
  \lmodel(T) \state  = q, \quad T = \exp(x), \quad {\rm and} \quad \mout = \omodel
  \state
  \label{eq:forward_e_d}
\end{equation}
where the discretized state $\state$ has dimension $m = 4961$, while
the transmissivity field $T$ and parameter $\param$ are of dimension
$n = 4800$.
Here, the matrix $\lmodel(T)$ can be expressed as
\[
\lmodel(T) = \sum_{i = 1}^{n} \lmodel_i \exp(x_i),
\]
where $\lmodel_i \in \real^{m \times m}$, $i = 1,\ldots,n$ are
parameter-independent matrices.

Model reduction for \eqref{eq:forward_e_d} consists of two
steps. Beginning with a reduced parameter basis $\pbasis_r$, we first
construct a DEIM interpolant for the log-normal process
$T \approx \exp( \pbasis_r \param_r ) $.
Given a basis $\fbasis_t$ that spans the subspace capturing variations
of $\exp( \pbasis_r \param_r )$, and the associated masking matrix
$\mask_t$, DEIM interpolation takes the form
\[
\exp( \pbasis_r \param_r ) \approx \fbasis_t^{} \alpha(\param_r^{}),
\quad {\rm where} \quad \alpha(\param_r^{}) = (\mask_t^\top
\fbasis_t^{} )^{-1} \exp( \mask_t^\top \pbasis_r^{} \param_r^{} ) .
\]
Using MATLAB notation, for a given reduced parameter $\param_r$, the
matrix $\lmodel(T)$ can be rewritten as
\begin{equation}
  \lmodel(T)  = \sum_{j = 1}^{t} \left( \sum_{i = 1}^{n} \lmodel_i^{} \fbasis_t^{}(i,j)
  \right) \alpha_j^{}(\param_r^{}),
  \label{eq:affine_e}
\end{equation}
where $\alpha_j(\param_r^{})$ is the $j$th component of the
vector-valued function $\alpha$.
In the second step, given a reduced state basis $\sbasis_s$, we
approximate the state by $\state \approx \sbasis_s \state_s$ and apply
Galerkin projection, yielding a reduced linear system
\[
\sbasis_s^\top \lmodel(T) \sbasis_s^{} \state_s^{} = \sbasis_s^\top q.
\]
Substituting \eqref{eq:affine_e} into the above equation, the reduced
order model can be written as
\[
\lmodel_s(\param_s) \state_s = q_s,
\]
where
\[
\lmodel_s(\param_s) = \sum_{j = 1}^{t} \sbasis_s^\top \left( \sum_{i =
    1}^{n} L_i^{} \fbasis_t^{}(i,j) \right) \sbasis_s^{}
\alpha_j^{}(\param_r^{}) \quad {\rm and} \quad q_s^{} = \sbasis_s^\top
q,
\]
and the associated reduced observation model is
$\mout = (\omodel \sbasis_s) \state_s$.
The computational cost of this reduced order model is dictated by the
dimension of the reduced parameter subspace, the reduced state
subspace, and the DEIM basis, and is independent of the dimension of
the original model.

\subsection{Numerical results}

We run the same set of tests as in the GOMOS example.
The Hellinger distance \eqref{eq:hell_p} is used to evaluate the
errors induced by various parameter reduction methods, versus the
number of basis vectors used in the parameter-approximated posterior.
The results are shown in Figure \ref{fig:elliptic_param_dh}.
As in the GOMOS example, the Laplace approximation (dashed line) has a
rather large Hellinger distance from the posterior.
Prior-KL parameter reduction (triangles) converges rather slowly
relative to the likelihood-informed methods; it outperforms the
Laplace approximation only after 25 or more basis vectors are
included.
Prior-LIPS (squares) converges more quickly than Prior-KL, but it is
outperformed by the other two likelihood-informed methods, which use
Hessians at parameter values drawn from approximations of the
posterior.
The convergence curves of Posterior-LIPS (crosses) and Laplace-LIPS
(diamonds) are almost identical.

\begin{figure}[h!]
  \centerline{\includegraphics[width=0.7\textwidth]{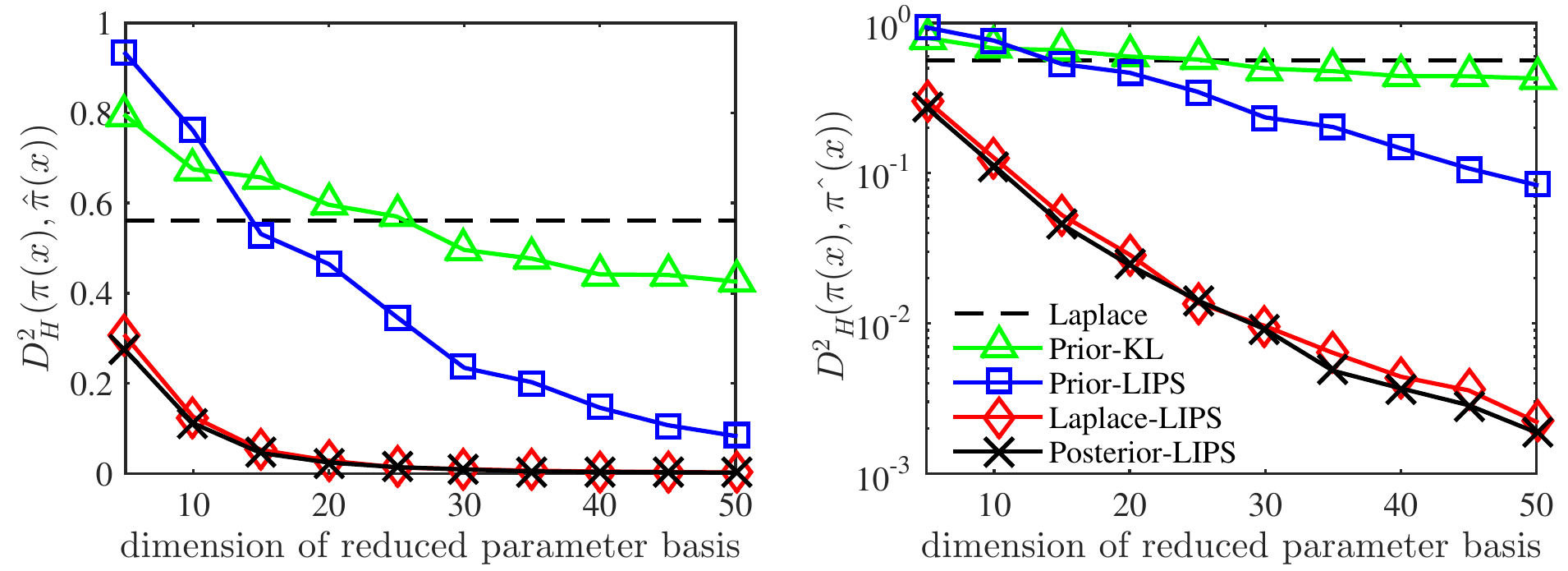}}
  \caption{Convergence of various \textit{parameter reduction}
    methods, for the elliptic PDE example: squared Hellinger distances
    \eqref{eq:hell_p} versus the dimensions of the reduced parameter
    bases $\pbasis_r$ defined by Posterior-LIPS, Laplace-LIPS,
    Prior-LIPS, and Prior-KL. The dashed line represents the squared
    Hellinger distance between the full posterior and its Laplace
    approximation \eqref{eq:lap}.}
  \label{fig:elliptic_param_dh}
\end{figure}

Given a 17-dimensional likelihood-informed parameter basis built via
Posterior-LIPS, we now use the Hellinger distance \eqref{eq:hell_m} to
evaluate the convergence of various state reduction methods, as a
function of the state dimension of the reduced-order model.
Because we have a non-smooth prior in this example, a rather
high-dimensional DEIM basis is required to ensure positivity in
approximating the log-normal process $\exp( \pbasis_r \param_r )$.
Four DEIM approximations of the log-normal process, with reduced bases
$\Theta_t$ having dimensions $250$, $350$, $450$, and $550$ are used
in this benchmark.
The results are shown in Figure \ref{fig:elliptic_rom_dh}.
In this test, Posterior-LISS is about one order of magnitude more
accurate than Laplace-LISS, and Laplace-LISS is about one order of
magnitude more accurate than Prior-POD.
We also note that Posterior-LISS generates the most stable convergence
curves among all the methods tested here.

\begin{figure}[h!]
  \centerline{\includegraphics[width=\textwidth]{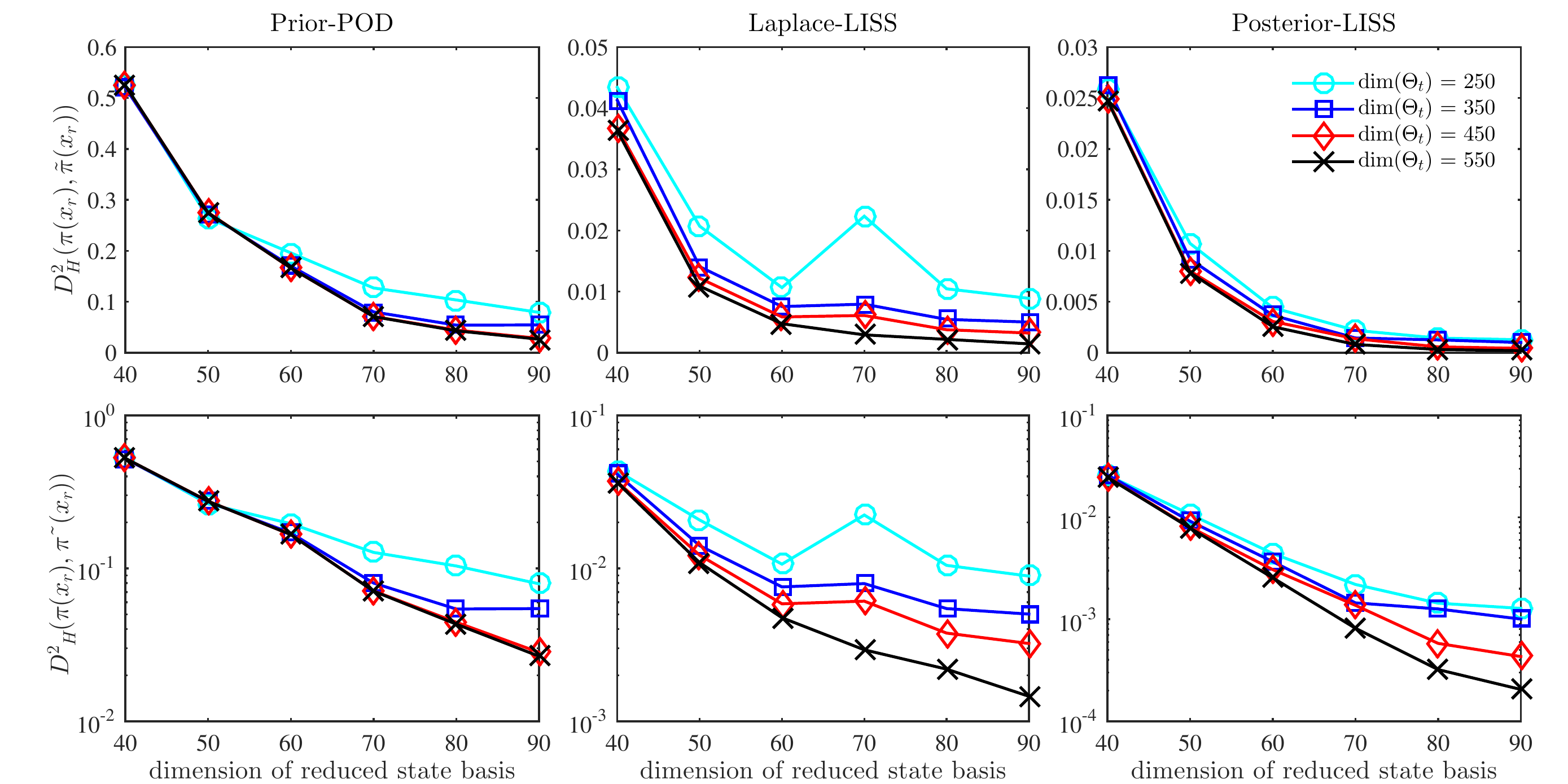}}
  \caption{Convergence of various \textit{state reduction} methods,
    for the elliptic PDE example. Squared Hellinger distances
    \eqref{eq:hell_m} versus the dimension of reduced state bases
    $\sbasis_s$ defined by Prior-POD, Laplace-LISS, and Posterior-LIS
    are shown in the left, middle, and right columns, respectively.
    Four different DEIM approximations of the exponential function
    $\exp( \pbasis_r \param_r )$, with dimensions $250$, $350$, $450$
    and $550$, are shown for each case. Again, errors are plotted with
    linear (top) and logarithmic (bottom) scales.}
  \label{fig:elliptic_rom_dh}
\end{figure}

We now demonstrate eight iterations of the \podi\ strategy (Algorithm
\ref{algo:post_joint}) using either the prior or the Laplace
approximation \eqref{eq:lap} as initial distributions.
The reduced parameter basis is truncated to retain components above
the eigenvalue threshold $\tau_{\rm g} = 0.05$, the DEIM basis
$\fbasis_t$ for interpolating the function
$\exp( \pbasis_r \param_r )$ is truncated to retain eigenvalues above
$10^{-14}$, and the reduced state basis is truncated to retain
eigenvalues above $10^{-6}$.
\updated{In this example, the numbers of forward model and GNH--action
  evaluations used in each iteration of \podi, or in a full run of
  \lapdi \ or \prdi, are the same as those used in the remote sensing
  example.
  We note that the action of the GNH in this case can be computed much
  more quickly than in the remote sensing case: the elliptic PDE
  forward model is self-adjoint, and hence the solution of the linear
  system in the forward model can be recycled to compute the action of
  the GNH without additional linear solves.
  Using the prior mean as the initial guess, the subspace trust region
  method requires 21 inexact Newton iterations to find the MAP. Each
  Newton iteration involves one forward model evaluation and the
  computation of the GNH action on an average of 5 directions. }

\begin{figure}[h!]
  \centerline{\includegraphics[width=0.7\textwidth]{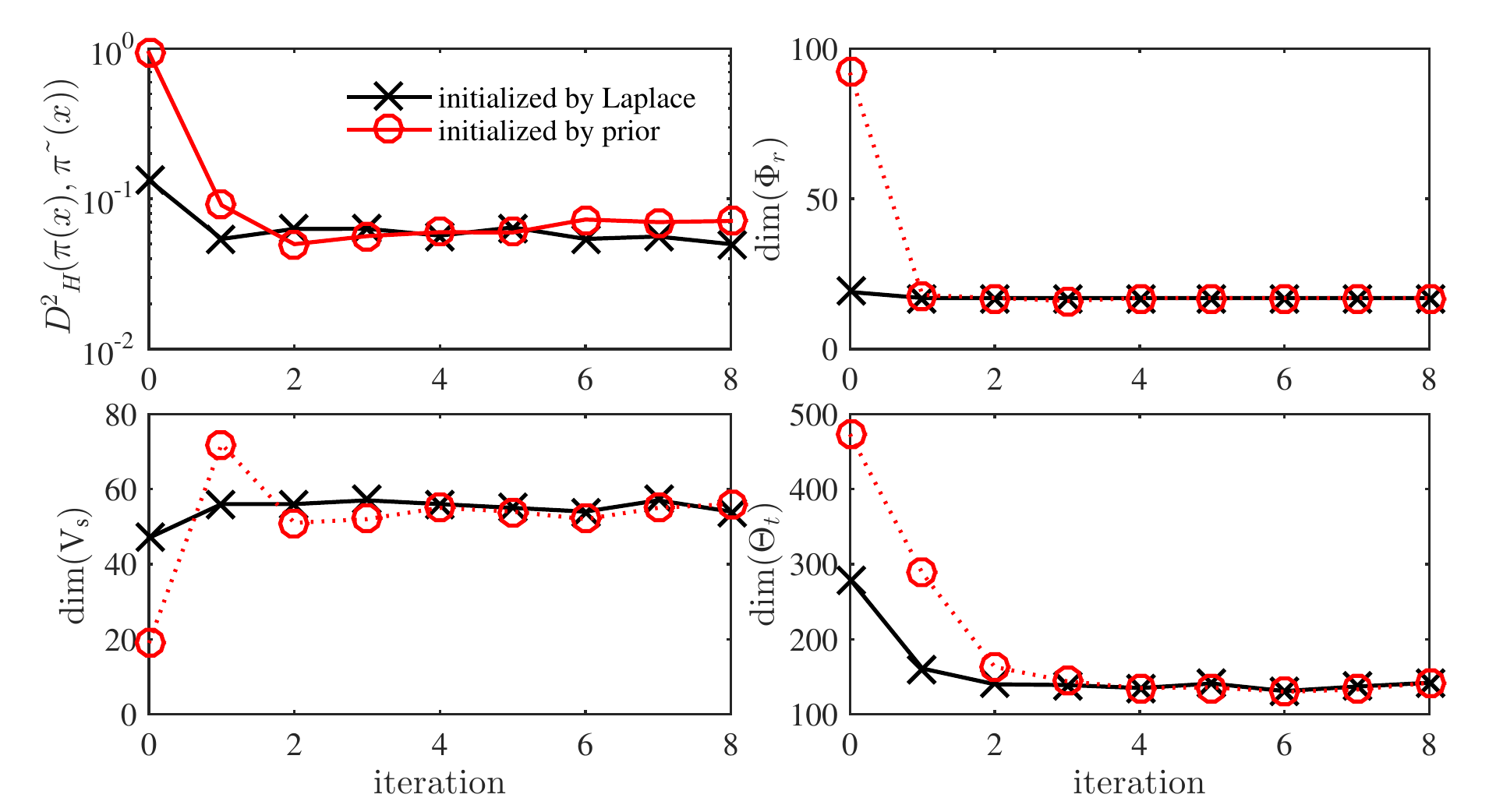}}
  \caption{Iterations of Algorithm \ref{algo:post_joint}, for the
    elliptic PDE example. For each iteration, we show the logarithm of
    the squared Hellinger distance \eqref{eq:hell_j}, and the
    dimension of the reduced parameter basis $\pbasis_r$, the reduced
    state basis $\sbasis_s$, and the DEIM basis $\fbasis_t$. We show
    iteration histories beginning with two different initial sampling
    distributions: the prior and the Laplace approximation
    \eqref{eq:lap}.}
  \label{fig:elliptic_joint}
\end{figure}

As in the GOMOS case, the Hellinger distance \eqref{eq:hell_j} from
the full posterior is used to evaluate the accuracy of the
jointly-approximated posterior. We also show the evolution of the
dimension of the reduced parameter basis $\pbasis_r$, the dimension of
the DEIM basis $\fbasis_t$, and the dimension of the reduced state
basis $\sbasis_s$. Results are shown in Figure
\ref{fig:elliptic_joint}. In this example, when the Laplace
approximation is used to initiate parameter sampling, the Hellinger
distance \eqref{eq:hell_j} drops slightly at iteration 1, and then
stays almost at a constant level for the remaining iterations.
The dimensions of various reduced bases also stabilize after the first
iteration. When the prior distribution is used to initiate parameter
sampling, the algorithm stabilizes at iteration 2, but the Hellinger
distance \eqref{eq:hell_j} is rather large in the first few iterations
($> 0.9$). We recall that the first iteration of Algorithm
\ref{algo:post_joint} generates a jointly-approximated posterior from
either \prdi~or \lapdi, depending on the initial distribution.
In Figure \ref{fig:elliptic_marginal}, we plot the full posterior and
the jointly-approximated posteriors generated by \podi, \lapdi, and
\prdi, as well as the Laplace approximation \eqref{eq:lap},
marginalized onto the first four KL bases.
Here the full posterior is sampled by the DILI method of
\cite{MCMC:CLM_2014}.
Both \lapdi\ and \podi\ have marginal distributions that are almost
identical to those of the full posterior, whereas the marginals of
\prdi~demonstrate significant discrepancies with the full posterior.

Overall, for this example, running \lapdi~is the most computationally
efficient way to construct the jointly-approximated posterior.
Beginning with this approximation and running an additional iteration
of \podi, the accuracy of the jointly-approximated posterior can be
further improved.

\begin{figure}[h!]
  \centerline{\includegraphics[width=\textwidth]{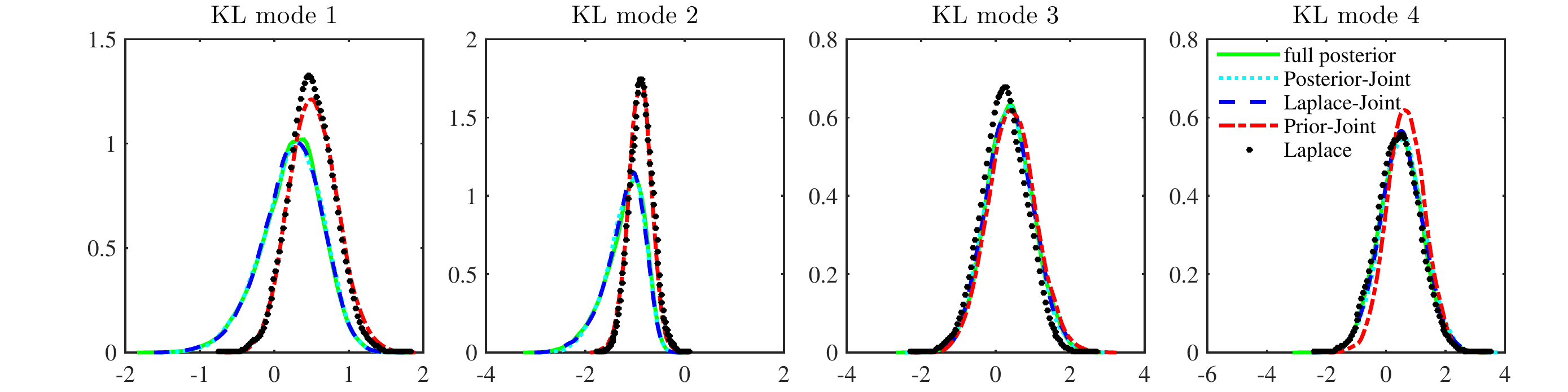}}
  \caption{Marginal posterior distributions for the elliptic PDE
    example: the full posterior and the jointly-approximated
    posteriors generated by \podi, \lapdi, and \prdi, along with the
    Laplace approximation \eqref{eq:lap}. We show marginal
    distributions along each of the first four KL modes.}
  \label{fig:elliptic_marginal}
\end{figure}


%
\section{Conclusions}
\label{sec:conclusions}

This paper addresses two major computational challenges in the
Bayesian solution of inverse problems: the high dimensionality of 
parameters and the large computational cost of forward model
evaluations. For many MCMC methods, the cost of sampling scales poorly
with the former, while the latter makes repeated evaluations of the
posterior density computationally prohibitive. We thus propose a
likelihood-informed approach for identifying and exploiting
low-dimensional structure in both the parameter space and the model
state space. The resulting \textit{jointly-approximated posterior} can 
be characterized using only low-dimensional samplers and inexpensive
evaluations of a reduced-order model. The computational cost of
computing posterior expectations then scales with the dimension of
the reduced parameter and state spaces, which reflect the intrinsic
complexities of the inverse problem---e.g., how large is the parameter
subspace capturing significant changes from prior to posterior, and
what variations in the forward model state are induced by the parameter
distribution on this subspace.

Previous work on the approximation of Bayesian inverse problems has
shown that it is useful to perform some kind of parameter reduction
before model reduction \cite{SuMo:MarNa_2009, ROM:LWG_2010}. It has
also been demonstrated that focusing attention on a particular region
of the parameter space---in particular, the region of high posterior
probability---enables the construction of more accurate reduced-order
models for the purpose of Bayesian inference \cite{ROM:CMW_2014}. The
present work uses locality in \textit{both} of these senses. We reduce
parameter dimension by ``filtering out'' directions of prior
variability that are irrelevant to the prior-to-posterior update,
i.e., directions that the likelihood does not inform. Within the
remaining parameter subspace, we focus on the region of parameter
values that is consistent with the data, i.e., that has high posterior
probability. Locality in dimension and in parameter range then yields
smaller variations in the model state, which are better captured by a
reduced-order model. The resulting jointly-approximated posterior can
be constructed for systems with high-dimensional parameters and
states, and is {accurate} in the stringent sense of Hellinger
distance from the true posterior.
As a byproduct of this reduction procedure, we are also able to
reduce the cost of handling high-dimensional data, by exploiting
low-dimensional structure and sparsity in the data space.

Within this framework, we introduce several alternative strategies for
constructing two key building blocks of the posterior approximation:
the likelihood-informed parameter subspace (LIPS) and the associated
likelihood-informed state subspace (LISS). One interesting aspect of
these strategies is the choice of reference distribution; options
include the prior, a Laplace approximation of the posterior, or the
posterior itself. Our numerical examples demonstrate that
posterior-focused reference distributions, while ``online'' in the
sense of depending on the data, yield the most accurate approximations
for a given subspace dimension. All of these strategies contain many
highly parallelizable computations.

While the present work has used snapshot-style approaches to
constructing appropriate subspaces of the parameter space and state
space, future work might use optimization on matrix manifolds to
directly search for optimal bases in nonlinear settings. It may also
be useful to extend beyond the essentially linear dimension reduction
strategies employed here to identify nonlinear manifolds that better
capture the prior-to-posterior update and associated variations of the
model state.
\updated{Finally, we note that our current approximations focus on
  inverse problems with nonlinear forward models but prescribed
  Gaussian priors. It is natural to consider generalizations to
  hierarchical Gaussian priors, where the prior is not precisely fixed
  but rather its mean and precision/covariance are controlled by
  additional hyperparameters. The posterior approximations developed
  here may be useful building blocks in this hierarchical Bayesian
  setting. }

%

\section*{Acknowledgments}

We thank Marko Laine and Johanna Tamminen from the Finnish
Meteorological Institute for providing us with the GOMOS model that
served as the baseline for our implementation of the remote sensing
example.
This work was supported by the US Department of Energy, Office of
Advanced Scientific Computing (ASCR), under grant number
DE-SC0009297.

\section*{References}

\bibliographystyle{model1-num-names}
\bibliography{bib_2015}

\begin{thebibliography}{60}
\expandafter\ifx\csname natexlab\endcsname\relax\def\natexlab#1{#1}\fi
\providecommand{\bibinfo}[2]{#2}
\ifx\xfnm\relax \def\xfnm[#1]{\unskip,\space#1}\fi
\bibitem[{Tarantola(2005)}]{IP:Tarantola_2004}
\bibinfo{author}{A.~Tarantola}, \bibinfo{title}{Inverse Problem Theory and
  Methods for Model Parameter Estimation}, \bibinfo{publisher}{Society for
  Industrial Mathematics}, \bibinfo{address}{Philadelphia},
  \bibinfo{year}{2005}.
\bibitem[{Kaipio and Somersalo(2004)}]{IP:KaiSo_2005}
\bibinfo{author}{J.~P. Kaipio}, \bibinfo{author}{E.~Somersalo},
  \bibinfo{title}{Statistical and Computational Inverse Problems}, volume
  \bibinfo{volume}{160}, \bibinfo{publisher}{Springer}, \bibinfo{address}{New
  York}, \bibinfo{year}{2004}.
\bibitem[{Stuart(2010)}]{IP:Stuart_2010}
\bibinfo{author}{A.~M. Stuart},
\newblock \bibinfo{title}{Inverse problems: a {B}ayesian perspective},
\newblock \bibinfo{journal}{Acta Numerica} \bibinfo{volume}{19}
  (\bibinfo{year}{2010}) \bibinfo{pages}{451--559}.
\bibitem[{Gilks et~al.(1996)Gilks, Richardson, and
  Spiegelhalter}]{MCMC:GRP_1996}
\bibinfo{editor}{W.~R. Gilks}, \bibinfo{editor}{S.~Richardson},
  \bibinfo{editor}{D.~J. Spiegelhalter} (Eds.), \bibinfo{title}{{M}arkov Chain
  {M}onte {C}arlo in practice}, volume~\bibinfo{volume}{2},
  \bibinfo{publisher}{CRC press}, \bibinfo{year}{1996}.
\bibitem[{Liu(2001)}]{MCMC:Liu_2001}
\bibinfo{author}{J.~S. Liu}, \bibinfo{title}{{M}onte {C}arlo strategies in
  Scientific Computing}, \bibinfo{publisher}{Springer}, \bibinfo{address}{New
  York}, \bibinfo{year}{2001}.
\bibitem[{Brooks et~al.(2011)Brooks, Gelman, Jones, and Meng}]{MCMC:BGJM_2011}
\bibinfo{editor}{S.~Brooks}, \bibinfo{editor}{A.~Gelman},
  \bibinfo{editor}{G.~Jones}, \bibinfo{editor}{X.~L. Meng} (Eds.),
  \bibinfo{title}{Handbook of {M}arkov Chain {M}onte {C}arlo},
  \bibinfo{publisher}{Taylor \& Francis}, \bibinfo{year}{2011}.
\bibitem[{Metropolis et~al.(1953)Metropolis, Rosenbluth, Rosenbluth, Teller,
  and Teller}]{MCMC:Metropolis_etal_1953}
\bibinfo{author}{N.~Metropolis}, \bibinfo{author}{A.~W. Rosenbluth},
  \bibinfo{author}{M.~N. Rosenbluth}, \bibinfo{author}{A.~H. Teller},
  \bibinfo{author}{E.~Teller},
\newblock \bibinfo{title}{Equation of state calculations by fast computing
  machines},
\newblock \bibinfo{journal}{Journal of Chemical Physics} \bibinfo{volume}{21}
  (\bibinfo{year}{1953}) \bibinfo{pages}{1087--1092}.
\bibitem[{Hastings(1970)}]{MCMC:Hastings_1970}
\bibinfo{author}{W.~Hastings},
\newblock \bibinfo{title}{{M}onte {C}arlo sampling using {M}arkov chains and
  their applications},
\newblock \bibinfo{journal}{Biometrika} \bibinfo{volume}{57}
  (\bibinfo{year}{1970}) \bibinfo{pages}{97--109}.
\bibitem[{Roberts et~al.(1997)Roberts, Gelman, and Gilks}]{MCMC:RGG_1997}
\bibinfo{author}{G.~O. Roberts}, \bibinfo{author}{A.~Gelman},
  \bibinfo{author}{W.~R. Gilks},
\newblock \bibinfo{title}{Weak convergence and optimal scaling of random walk
  {M}etropolis algorithms},
\newblock \bibinfo{journal}{Annals of Applied Probability} \bibinfo{volume}{7}
  (\bibinfo{year}{1997}) \bibinfo{pages}{110--120}.
\bibitem[{Roberts and Rosenthal(1998)}]{MCMC:RoRo_1998}
\bibinfo{author}{G.~O. Roberts}, \bibinfo{author}{J.~S. Rosenthal},
\newblock \bibinfo{title}{Optimal scaling of discrete approximations to
  {L}angevin diffusions},
\newblock \bibinfo{journal}{Journal of the Royal Statistical Society: Series B
  (Statistical Methodology)} \bibinfo{volume}{60} (\bibinfo{year}{1998})
  \bibinfo{pages}{255--268}.
\bibitem[{Roberts and Rosenthal(2001)}]{MCMC:RoRo_2001}
\bibinfo{author}{G.~O. Roberts}, \bibinfo{author}{J.~S. Rosenthal},
\newblock \bibinfo{title}{Optimal scaling of various {M}etropolis-{H}astings
  algorithms},
\newblock \bibinfo{journal}{Statistical Science} \bibinfo{volume}{16}
  (\bibinfo{year}{2001}) \bibinfo{pages}{351--367}.
\bibitem[{Mattingly et~al.(2012)Mattingly, Pillai, and Stuart}]{MCMC:MPS_2012}
\bibinfo{author}{J.~C. Mattingly}, \bibinfo{author}{N.~Pillai},
  \bibinfo{author}{A.~M. Stuart},
\newblock \bibinfo{title}{Diffusion limits of the random walk {M}etropolis
  algorithm in high dimensions},
\newblock \bibinfo{journal}{Annals of Applied Probability} \bibinfo{volume}{22}
  (\bibinfo{year}{2012}) \bibinfo{pages}{881--930}.
\bibitem[{Pillai et~al.(2012)Pillai, Stuart, and Thiery}]{MCMC:PST_2012}
\bibinfo{author}{N.~S. Pillai}, \bibinfo{author}{A.~M. Stuart},
  \bibinfo{author}{A.~H. Thiery},
\newblock \bibinfo{title}{{Optimal scaling and diffusion limits for the
  {L}angevin algorithm in high dimensions}},
\newblock \bibinfo{journal}{Annals of Applied Probability} \bibinfo{volume}{22}
  (\bibinfo{year}{2012}) \bibinfo{pages}{2320--2356}.
\bibitem[{Spantini et~al.(2015)Spantini, Solonen, Cui, Martin, Tenorio, and
  Marzouk}]{DimRedu:Spantini_etal_2015}
\bibinfo{author}{A.~Spantini}, \bibinfo{author}{A.~Solonen},
  \bibinfo{author}{T.~Cui}, \bibinfo{author}{J.~Martin},
  \bibinfo{author}{L.~Tenorio}, \bibinfo{author}{Y.~M. Marzouk},
\newblock \bibinfo{title}{Optimal low-rank approximation of linear {B}ayesian
  inverse problems},
\newblock \bibinfo{journal}{SIAM Journal on Scientific Computing}
  \bibinfo{volume}{37} (\bibinfo{year}{2015}) \bibinfo{pages}{A2451--A2487}.
\bibitem[{Cui et~al.(2014{\natexlab{a}})Cui, Martin, Marzouk, Solonen, and
  Spantini}]{DimRedu:Cui_etal_2014}
\bibinfo{author}{T.~Cui}, \bibinfo{author}{J.~Martin}, \bibinfo{author}{Y.~M.
  Marzouk}, \bibinfo{author}{A.~Solonen}, \bibinfo{author}{A.~Spantini},
\newblock \bibinfo{title}{Likelihood-informed dimension reduction for nonlinear
  inverse problems},
\newblock \bibinfo{journal}{Inverse Problems} \bibinfo{volume}{30}
  (\bibinfo{year}{2014}{\natexlab{a}}) \bibinfo{pages}{114015}.
\bibitem[{Cui et~al.(2014{\natexlab{b}})Cui, Marzouk, and
  Willcox}]{ROM:CMW_2014}
\bibinfo{author}{T.~Cui}, \bibinfo{author}{Y.~M. Marzouk},
  \bibinfo{author}{K.~E. Willcox},
\newblock \bibinfo{title}{Data-driven model reduction for the {B}ayesian
  solution of inverse problems},
\newblock \bibinfo{journal}{Int. J. Numer. Meth. Engng} \bibinfo{volume}{102}
  (\bibinfo{year}{2014}{\natexlab{b}}) \bibinfo{pages}{966--990}.
\bibitem[{Cui et~al.(2016)Cui, Law, and Marzouk}]{MCMC:CLM_2014}
\bibinfo{author}{T.~Cui}, \bibinfo{author}{K.~J.~H. Law},
  \bibinfo{author}{Y.~M. Marzouk},
\newblock \bibinfo{title}{Dimension-independent likelihood-informed {MCMC}},
\newblock \bibinfo{journal}{Journal of Computational Physics}
  \bibinfo{volume}{304} (\bibinfo{year}{2016}) \bibinfo{pages}{109--137}.
\bibitem[{Noor et~al.(1981)Noor, Andersen, and Peters}]{ROM:NAP_1981}
\bibinfo{author}{A.~Noor}, \bibinfo{author}{C.~Andersen},
  \bibinfo{author}{J.~Peters},
\newblock \bibinfo{title}{Reduced basis technique for collapse analysis of
  shells},
\newblock \bibinfo{journal}{AIAA Journal} \bibinfo{volume}{19}
  (\bibinfo{year}{1981}) \bibinfo{pages}{393--397}.
\bibitem[{Sirovich(1987)}]{ROM:Sirovich_1987}
\bibinfo{author}{L.~Sirovich},
\newblock \bibinfo{title}{Turbulence and the dynamics of coherent structures.
  part 1: Coherent structures},
\newblock \bibinfo{journal}{Quarterly of Applied Mathematics}
  \bibinfo{volume}{45} (\bibinfo{year}{1987}) \bibinfo{pages}{561--571}.
\bibitem[{Holmes et~al.(1996)Holmes, Lumley, and Berkooz}]{ROM:HLB_1996}
\bibinfo{author}{P.~Holmes}, \bibinfo{author}{J.~Lumley},
  \bibinfo{author}{G.~Berkooz}, \bibinfo{title}{Turbulence, coherent
  structures, and dynamical systems and symmetry},
  \bibinfo{publisher}{Cambridge University Press},
  \bibinfo{address}{Cambridge}, \bibinfo{year}{1996}.
\bibitem[{Patera and Rozza(2007)}]{ROM:PaRo_2007}
\bibinfo{author}{A.~T. Patera}, \bibinfo{author}{G.~Rozza},
  \bibinfo{title}{Reduced basis approximation and a posteriori error estimation
  for parametrized partial differential equations}, MIT Pappalardo monographs
  in mechanical engineering, \bibinfo{publisher}{Copyright MIT (2006--2007)},
  \bibinfo{year}{2007}.
\bibitem[{Christen and Fox(2005)}]{MCMC:ChriFox_2005}
\bibinfo{author}{J.~A. Christen}, \bibinfo{author}{C.~Fox},
\newblock \bibinfo{title}{{MCMC} using an approximation},
\newblock \bibinfo{journal}{Journal of Computational and Graphical statistics}
  \bibinfo{volume}{14} (\bibinfo{year}{2005}) \bibinfo{pages}{795--810}.
\bibitem[{Cui(2010)}]{MCMC:Cui_2010}
\bibinfo{author}{T.~Cui}, \bibinfo{title}{Bayesian calibration of geothermal
  reservoir models via Markov Chain Monte Carlo}, Ph.D. thesis, The University
  of Auckland, \bibinfo{year}{2010}.
\bibitem[{Marzouk and Najm(2009)}]{SuMo:MarNa_2009}
\bibinfo{author}{Y.~M. Marzouk}, \bibinfo{author}{H.~N. Najm},
\newblock \bibinfo{title}{Dimensionality reduction and polynomial chaos
  acceleration of {B}ayesian inference in inverse problems},
\newblock \bibinfo{journal}{Journal of Computational Physics}
  \bibinfo{volume}{228} (\bibinfo{year}{2009}) \bibinfo{pages}{1862--1902}.
\bibitem[{Karhunen(1947)}]{DimRedu:Karhunen_1947}
\bibinfo{author}{K.~Karhunen},
\newblock \bibinfo{title}{\"{U}ber lineare methoden in der
  wahrscheinlichkeitsrechnung},
\newblock \bibinfo{journal}{Ann. Acad. Sci. Fennicae. Ser. A. I. Math.-Phys}
  \bibinfo{volume}{37} (\bibinfo{year}{1947}) \bibinfo{pages}{1--79}.
\bibitem[{Lo\`{e}ve(1978)}]{DimRedu:Loeve_1978}
\bibinfo{author}{M.~Lo\`{e}ve}, \bibinfo{title}{Probability theory, Vol. II},
  volume~\bibinfo{volume}{46} of \textit{\bibinfo{series}{Graduate Texts in
  Mathematics}}, \bibinfo{publisher}{Springer-Verlag},
  \bibinfo{address}{Berlin}, \bibinfo{edition}{4} edition,
  \bibinfo{year}{1978}.
\bibitem[{Ghanem and Spanos(1991)}]{SuMo:GhaSpa_1991}
\bibinfo{author}{R.~Ghanem}, \bibinfo{author}{P.~Spanos},
  \bibinfo{title}{Stochastic Finite Elements: A Spectral Approach},
  \bibinfo{publisher}{Springer-Verlag}, \bibinfo{address}{New York},
  \bibinfo{year}{1991}.
\bibitem[{Xiu and Karniadakis(2002)}]{SuMo:XiuKar_2002}
\bibinfo{author}{D.~Xiu}, \bibinfo{author}{G.~E. Karniadakis},
\newblock \bibinfo{title}{The {W}iener-{A}skey polynomial chaos for stochastic
  differential equations},
\newblock \bibinfo{journal}{SIAM Journal on Scientific Computing}
  \bibinfo{volume}{24} (\bibinfo{year}{2002}) \bibinfo{pages}{619--644}.
\bibitem[{Xiu(2010)}]{SuMo:Xiu_2010}
\bibinfo{author}{D.~Xiu}, \bibinfo{title}{Numerical Methods for Stochastic
  Computations: A Spectral Method Approach}, \bibinfo{publisher}{Princeton
  University Press}, \bibinfo{address}{Princeton}, \bibinfo{year}{2010}.
\bibitem[{Lipponen et~al.(2013)Lipponen, Sepp\"{a}nen, and
  Kaipio}]{ROM:LSK_2013}
\bibinfo{author}{A.~Lipponen}, \bibinfo{author}{A.~Sepp\"{a}nen},
  \bibinfo{author}{J.~P. Kaipio},
\newblock \bibinfo{title}{Electrical impedance tomography imaging with
  reduced-order model based on proper orthogonal decomposition},
\newblock \bibinfo{journal}{Journal of Electronic Imaging} \bibinfo{volume}{22}
  (\bibinfo{year}{2013}) \bibinfo{pages}{023008}.
\bibitem[{Ma and Zabaras(2009)}]{SuMo:MaZa_2009}
\bibinfo{author}{X.~Ma}, \bibinfo{author}{N.~Zabaras},
\newblock \bibinfo{title}{An efficient data-driven {B}ayesian inference
  approach to inverse problems based on adaptive sparse grid collocation
  method},
\newblock \bibinfo{journal}{Inverse Problems} \bibinfo{volume}{25}
  (\bibinfo{year}{2009}) \bibinfo{pages}{035013}.
\bibitem[{Higdon(2002)}]{Sto:Higdon_2002}
\bibinfo{author}{D.~Higdon},
\newblock \bibinfo{title}{Space and space-time modeling using process
  convolutions},
\newblock in: \bibinfo{editor}{C.~Anderson}, \bibinfo{editor}{V.~Barnett},
  \bibinfo{editor}{P.~C. Chatwin}, \bibinfo{editor}{A.~H. El-Shaarawi} (Eds.),
  \bibinfo{booktitle}{Quantitative methods for current environmental issues},
  \bibinfo{publisher}{Springer}, \bibinfo{year}{2002}, pp.
  \bibinfo{pages}{37--56}.
\bibitem[{Lieberman et~al.(2010)Lieberman, Willcox, and Ghattas}]{ROM:LWG_2010}
\bibinfo{author}{C.~Lieberman}, \bibinfo{author}{K.~E. Willcox},
  \bibinfo{author}{O.~Ghattas},
\newblock \bibinfo{title}{Parameter and state model reduction for large-scale
  statistical inverse problems},
\newblock \bibinfo{journal}{SIAM Journal on Scientific Computing}
  \bibinfo{volume}{32} (\bibinfo{year}{2010}) \bibinfo{pages}{2523--2542}.
\bibitem[{Flath et~al.(2011)Flath, Wilcox, Akcelik, Hill, Waanders, and
  Ghattas}]{IP:Flath_etal_2011}
\bibinfo{author}{H.~P. Flath}, \bibinfo{author}{L.~C. Wilcox},
  \bibinfo{author}{V.~Akcelik}, \bibinfo{author}{J.~Hill},
  \bibinfo{author}{B.~V.~B. Waanders}, \bibinfo{author}{O.~Ghattas},
\newblock \bibinfo{title}{Fast algorithms for {B}ayesian uncertainty
  quantification in large-scale linear inverse problems based on low-rank
  partial {H}essian approximations},
\newblock \bibinfo{journal}{SIAM Journal on Scientific Computing}
  \bibinfo{volume}{33} (\bibinfo{year}{2011}) \bibinfo{pages}{407--432}.
\bibitem[{Martin et~al.(2012)Martin, Wilcox, Burstedde, and
  Ghattas}]{MCMC:MWBG_2012}
\bibinfo{author}{J.~Martin}, \bibinfo{author}{L.~C. Wilcox},
  \bibinfo{author}{C.~Burstedde}, \bibinfo{author}{O.~Ghattas},
\newblock \bibinfo{title}{A stochastic {N}ewton {MCMC} method for large-scale
  statistical inverse problems with application to seismic inversion},
\newblock \bibinfo{journal}{SIAM Journal on Scientific Computing}
  \bibinfo{volume}{34} (\bibinfo{year}{2012}) \bibinfo{pages}{A1460--A1487}.
\bibitem[{Bui-Thanh et~al.(2013)Bui-Thanh, Ghattas, Martin, and
  Stadler}]{IP:BGMS_2013}
\bibinfo{author}{T.~Bui-Thanh}, \bibinfo{author}{O.~Ghattas},
  \bibinfo{author}{J.~Martin}, \bibinfo{author}{G.~Stadler},
\newblock \bibinfo{title}{A computational framework for infinite-dimensional
  {B}ayesian inverse problems. {P}art {I}: {T}he linearized case, with
  application to global seismic inversion},
\newblock \bibinfo{journal}{SIAM Journal on Scientific Computing}
  \bibinfo{volume}{35} (\bibinfo{year}{2013}) \bibinfo{pages}{A2494--A2523}.
\bibitem[{Petra et~al.(2014)Petra, Martin, Stadler, and
  Ghattas}]{MCMC:Petra_etal_2014}
\bibinfo{author}{N.~Petra}, \bibinfo{author}{J.~Martin},
  \bibinfo{author}{G.~Stadler}, \bibinfo{author}{O.~Ghattas},
\newblock \bibinfo{title}{A computational framework for infinite-dimensional
  {B}ayesian inverse problems: Part {II.} {S}tochastic {N}ewton {MCMC} with
  application to ice sheet flow inverse problems},
\newblock \bibinfo{journal}{SIAM Journal on Scientific Computing}
  \bibinfo{volume}{36} (\bibinfo{year}{2014}) \bibinfo{pages}{A1525--A1555}.
\bibitem[{F\"{o}rstner and Boudewijn(2003)}]{Lin:Forbou_2003}
\bibinfo{author}{W.~F\"{o}rstner}, \bibinfo{author}{M.~Boudewijn},
\newblock \bibinfo{title}{A metric for covariance matrices},
\newblock in: \bibinfo{booktitle}{Geodesy: The Challenge of the 3rd
  Millennium}, \bibinfo{publisher}{Springer}, \bibinfo{address}{Berlin
  Heidelberg}, \bibinfo{year}{2003}, pp. \bibinfo{pages}{299--309}.
\bibitem[{Parlett(1980)}]{Lin:Parlett_1980}
\bibinfo{author}{B.~N. Parlett}, \bibinfo{title}{The symmetric eigenvalue
  problem}, \bibinfo{publisher}{Prentice-Hall}, \bibinfo{year}{1980}.
\bibitem[{Astrid et~al.(2008)Astrid, Weiland, Willcox, and
  Backx}]{ROM:AWWB_2008}
\bibinfo{author}{P.~Astrid}, \bibinfo{author}{S.~Weiland},
  \bibinfo{author}{K.~E. Willcox}, \bibinfo{author}{T.~Backx},
\newblock \bibinfo{title}{Missing point estimation in models described by
  proper orthogonal decomposition},
\newblock \bibinfo{journal}{IEEE Transactions on Automatic Control}
  \bibinfo{volume}{53} (\bibinfo{year}{2008}) \bibinfo{pages}{2237--2251.}
\bibitem[{Barrault et~al.(2004)Barrault, Maday, Nguyen, and
  Patera}]{ROM:BMNP_2004}
\bibinfo{author}{M.~Barrault}, \bibinfo{author}{Y.~Maday},
  \bibinfo{author}{N.~C. Nguyen}, \bibinfo{author}{A.~T. Patera},
\newblock \bibinfo{title}{An ``empirical interpolation'' method: application to
  efficient reduced-basis discretization of partial differential equations},
\newblock \bibinfo{journal}{Comptes Rendus Mathematique} \bibinfo{volume}{339}
  (\bibinfo{year}{2004}) \bibinfo{pages}{667--672.}
\bibitem[{Chaturantabut and Sorensen(2010)}]{ROM:ChaSor_2010}
\bibinfo{author}{S.~Chaturantabut}, \bibinfo{author}{D.~C. Sorensen},
\newblock \bibinfo{title}{Nonlinear model reduction via discrete empirical
  interpolation},
\newblock \bibinfo{journal}{SIAM Journal on Scientific Computing}
  \bibinfo{volume}{32} (\bibinfo{year}{2010}) \bibinfo{pages}{2737--2764.}
\bibitem[{Wang and Zabaras(2005)}]{ROM:WangZab_2005}
\bibinfo{author}{J.~Wang}, \bibinfo{author}{N.~Zabaras},
\newblock \bibinfo{title}{Using {B}ayesian statistics in the estimation of heat
  source in radiation},
\newblock \bibinfo{journal}{International Journal of Heat and Mass Transfer}
  \bibinfo{volume}{48} (\bibinfo{year}{2005}) \bibinfo{pages}{15--29}.
\bibitem[{Galbally et~al.(2008)Galbally, Fidkowski, Willcox, and
  Ghattas}]{ROM:GFWG_2008}
\bibinfo{author}{D.~Galbally}, \bibinfo{author}{K.~Fidkowski},
  \bibinfo{author}{K.~E. Willcox}, \bibinfo{author}{O.~Ghattas},
\newblock \bibinfo{title}{Nonlinear model reduction for uncertainty
  quantification in large scale inverse problems},
\newblock \bibinfo{journal}{International journal for numerical methods in
  engineering} \bibinfo{volume}{81} (\bibinfo{year}{2008})
  \bibinfo{pages}{1581--1608}.
\bibitem[{Golub and Loan(2012)}]{Lin:GoVanlo_2012}
\bibinfo{author}{G.~H. Golub}, \bibinfo{author}{C.~F.~V. Loan},
  \bibinfo{title}{Matrix Computations}, \bibinfo{publisher}{JHU Press},
  \bibinfo{year}{2012}.
\bibitem[{Halko et~al.(2011)Halko, Martinsson, and Tropp}]{Lin:HMT_2011}
\bibinfo{author}{N.~Halko}, \bibinfo{author}{P.~Martinsson},
  \bibinfo{author}{J.~A. Tropp},
\newblock \bibinfo{title}{Finding structure with randomness: Probabilistic
  algorithms for constructing approximate matrix decompositions},
\newblock \bibinfo{journal}{SIAM Review} \bibinfo{volume}{53}
  (\bibinfo{year}{2011}) \bibinfo{pages}{217--288}.
\bibitem[{Liberty et~al.(2007)Liberty, Woolfe, Martinsson, Rokhlin, and
  Tygert}]{Lin:Liberty_etal_2007}
\bibinfo{author}{E.~Liberty}, \bibinfo{author}{F.~Woolfe},
  \bibinfo{author}{P.~G. Martinsson}, \bibinfo{author}{V.~Rokhlin},
  \bibinfo{author}{M.~Tygert},
\newblock \bibinfo{title}{Randomized algorithms for the low-rank approximation
  of matrices},
\newblock \bibinfo{journal}{Proceedings of the National Academy of Sciences}
  \bibinfo{volume}{104} (\bibinfo{year}{2007}) \bibinfo{pages}{20167--20172}.
\bibitem[{Atchad\'{e}(2006)}]{MCMC:Atchade_2006}
\bibinfo{author}{Y.~F. Atchad\'{e}},
\newblock \bibinfo{title}{An adaptive version for the {M}etropolis adjusted
  {L}angevin algorithm with a truncated drift},
\newblock \bibinfo{journal}{Methodology and Computing in {A}pplied Probability}
  \bibinfo{volume}{8} (\bibinfo{year}{2006}) \bibinfo{pages}{235--254}.
\bibitem[{Chorin and Tu(2009)}]{Sto:ChoTu_2009}
\bibinfo{author}{A.~J. Chorin}, \bibinfo{author}{X.~Tu},
\newblock \bibinfo{title}{Implicit sampling for particle filters},
\newblock \bibinfo{journal}{Physical Review D} \bibinfo{volume}{106}
  (\bibinfo{year}{2009}) \bibinfo{pages}{17249--17254}.
\bibitem[{Morzfeld et~al.(2012)Morzfeld, Tu, Atkins, and
  Chorin}]{Sto:MTAC_2012}
\bibinfo{author}{M.~Morzfeld}, \bibinfo{author}{X.~Tu},
  \bibinfo{author}{E.~Atkins}, \bibinfo{author}{A.~J. Chorin},
\newblock \bibinfo{title}{A random map implementation of implicit filters},
\newblock \bibinfo{journal}{Journal of Computational Physics}
  \bibinfo{volume}{231} (\bibinfo{year}{2012}) \bibinfo{pages}{2049--2066}.
\bibitem[{Bardsley et~al.(2014)Bardsley, Solonen, Haario, and
  Laine}]{Sto:BSHL_2014}
\bibinfo{author}{J.~M. Bardsley}, \bibinfo{author}{A.~Solonen},
  \bibinfo{author}{H.~Haario}, \bibinfo{author}{M.~Laine},
\newblock \bibinfo{title}{Randomize-then-optimize: A method for sampling from
  posterior distributions in nonlinear inverse problems.},
\newblock \bibinfo{journal}{SIAM Journal on Scientific Computing}
  \bibinfo{volume}{36} (\bibinfo{year}{2014}) \bibinfo{pages}{A1895--A1910}.
\bibitem[{Schwab and Stuart(2012)}]{IP:SchStu_2012}
\bibinfo{author}{C.~Schwab}, \bibinfo{author}{A.~M. Stuart},
\newblock \bibinfo{title}{Sparse deterministic approximation of {B}ayesian
  inverse problems},
\newblock \bibinfo{journal}{Inverse Problems} \bibinfo{volume}{28}
  (\bibinfo{year}{2012}) \bibinfo{pages}{045003}.
\bibitem[{Schillings and Schwab(2013)}]{IP:SchiSch_2013}
\bibinfo{author}{C.~Schillings}, \bibinfo{author}{C.~Schwab},
\newblock \bibinfo{title}{Sparse, adaptive {S}molyak quadratures for {B}ayesian
  inverse problems},
\newblock \bibinfo{journal}{Inverse Problems} \bibinfo{volume}{29}
  (\bibinfo{year}{2013}) \bibinfo{pages}{065011}.
\bibitem[{Chen and Schwab(2015)}]{IP:ChenSch_2015}
\bibinfo{author}{P.~Chen}, \bibinfo{author}{C.~Schwab},
\newblock \bibinfo{title}{Sparse-grid, reduced-basis {B}ayesian inversion},
\newblock \bibinfo{journal}{Computer Methods in Applied Mechanics and
  Engineering} \bibinfo{volume}{297} (\bibinfo{year}{2015})
  \bibinfo{pages}{84--115}.
\bibitem[{Owen(2013)}]{Owen_2013}
\bibinfo{author}{A.~B. Owen}, \bibinfo{title}{Monte Carlo theory, methods and
  examples}, \bibinfo{year}{2013}.
\bibitem[{Haario et~al.(2004)Haario, Laine, Lehtinen, Saksman, and
  Tamminen}]{IP:Haario_etal_2004}
\bibinfo{author}{H.~Haario}, \bibinfo{author}{M.~Laine},
  \bibinfo{author}{M.~Lehtinen}, \bibinfo{author}{E.~Saksman},
  \bibinfo{author}{J.~Tamminen},
\newblock \bibinfo{title}{{M}arkov chain {M}onte {C}arlo methods for high
  dimensional inversion in remote sensing},
\newblock \bibinfo{journal}{Journal of the Royal Statistical Society: Series B
  (Statistical Methodology)} \bibinfo{volume}{66} (\bibinfo{year}{2004})
  \bibinfo{pages}{591--608}.
\bibitem[{Tamminen(2004)}]{MCMC:Tamminen_2004}
\bibinfo{author}{J.~Tamminen}, \bibinfo{title}{Adaptive Markov chain Monte
  Carlo algorithms with geophysical applications}, Ph.D. thesis, Finnish
  Meteorologocal Institute, \bibinfo{year}{2004}.
\bibitem[{Coleman and Li(1994)}]{Opt:CoLi_1994}
\bibinfo{author}{T.~F. Coleman}, \bibinfo{author}{Y.~Li},
\newblock \bibinfo{title}{On the convergence of interior-reflective newton
  methods for nonlinear minimization subject to bounds},
\newblock \bibinfo{journal}{Mathematical programming} \bibinfo{volume}{67}
  (\bibinfo{year}{1994}) \bibinfo{pages}{189--224}.
\bibitem[{Coleman and Li(1996)}]{Opt:CoLi_1996}
\bibinfo{author}{T.~F. Coleman}, \bibinfo{author}{Y.~Li},
\newblock \bibinfo{title}{An interior trust region approach for nonlinear
  minimization subject to bounds},
\newblock \bibinfo{journal}{SIAM Journal on optimization} \bibinfo{volume}{6}
  (\bibinfo{year}{1996}) \bibinfo{pages}{418--445}.
\bibitem[{Lindgren et~al.(2011)Lindgren, Rue, and Lindstr\"{o}m}]{MRF:LRL_2011}
\bibinfo{author}{F.~Lindgren}, \bibinfo{author}{H.~Rue},
  \bibinfo{author}{J.~Lindstr\"{o}m},
\newblock \bibinfo{title}{An explicit link between {G}aussian fields and
  {G}aussian {M}arkov random fields: the stochastic partial differential
  equation approach.},
\newblock \bibinfo{journal}{Journal of the Royal Statistical Society: Series B
  (Statistical Methodology)} \bibinfo{volume}{73} (\bibinfo{year}{2011})
  \bibinfo{pages}{423--498}.

\end{thebibliography}

\end{document}